\documentclass[conference]{IEEEtran}
\pdfoutput=0
\IEEEoverridecommandlockouts

\usepackage{changepage,cite}
\usepackage{amsmath,amssymb,amsfonts,amsthm}
\usepackage{tabularx,multirow,textcomp}
\usepackage{relsize,bbold}
\usepackage{graphicx,caption,subcaption,xcolor}
\usepackage{float,booktabs}
\usepackage{import}
\usepackage{algorithm,algpseudocode}

\def\ls{\mspace{4mu}}

\def\gmk{\boldsymbol{g}_{m,k}}
\def\gmkh{\boldsymbol{\hat{g}}_{m,k}}
\def\gmkt{\boldsymbol{\tilde{g}}_{m,k}}

\def\Gm{\boldsymbol{{G}}}
\def\Gmh{\boldsymbol{{\hat{G}}}}
\def\Gmt{\boldsymbol{{\tilde{G}}}}
\def\Pmat{\boldsymbol{{P}}}
\def\Ms{\boldsymbol{{M}}^{(s)}}
\def\Mi{\boldsymbol{{M}}^{(i)}}
\def\Tmat{\boldsymbol{{T}}}

\def\Wm{\boldsymbol{W}_{m}}

\def\Wmh{\boldsymbol{W}_{m}^*}

\def\x{\boldsymbol{x}}

\def\hmk{\boldsymbol{h}_{m,k}}

\makeatletter
\newcommand{\removelatexerror}{\let\@latex@error\@gobble}
\makeatother
\newtheorem{theorem}{Theorem}

\newtheorem{prop}{Proposition}

\begin{document}
\title{Smart Hybrid Beamforming and Pilot Assignment for 6G Cell-Free Massive MIMO}

\author{Carles Diaz-Vilor,$^{1}$ Alexei Ashikhmin,$^{2}$ and Hong Yang$^{2}$% <-this  stops a space
\thanks{$^{1}$C. Diaz-Vilor is with the Univ. of California, Irvine. Email: 
        {\tt\small \{cdiazvil\} at uci.edu }} %
\thanks{$^{2}$A. Ashikhmin and H. Yang are with Nokia Bell Labs, Murray Hill. Emails:
        {\tt\small  \{alexei.ashikhmin, h.yang\} at nokia-bell-labs.com }}
}

\maketitle
\begin{abstract}
We investigate Cell-Free massive MIMO networks, where each access point (AP) is equipped with a hybrid  analog-digital transceiver, reducing the complexity and cost compared to a fully digital transceiver. Asymptotic approximations for the spectral efficiency are derived for uplink and downlink. Capitalizing on these expressions, a max-min problem is formulated enabling us to optimize the (\emph{i}) analog beamformer  at the APs and (\emph{ii}) pilot assignment. Simulations show that the optimization of these variables substantially increases the minimum user  throughput
\end{abstract}


\IEEEpeerreviewmaketitle

\section{Introduction}
A prospective candidate considered for beyond-5G wireless networks is the cell-free massive MIMO (CF-mMIMO) topology, where every user (UE) potentially connects to every access point (AP), and takes the principles of cell cooperation to the limit; see \cite{7827017, 8845768, 7917284,8630677, 9043895} and the references therein.

In parallel, forthcoming technologies will be operating at higher frequencies (i.e. mmWave or THz bands), and therefore the transceivers complexity experiences a key trade-off: data rate vs power consumption. Additionally, CF networks will cover larger areas compared to cellular systems, and therefore the severity of the path loss requires the APs to be equipped with large arrays to compensate the attenuation, demanding even more power if fully digital structures are used. %Therefore, power consumption must be taken into account in the design of these networks.

A possible solution that has attracted a lot of attention is a hybrid transceiver \cite{7445130,6717211}, composed by two stages: (a) the analog part, in which the antennas are connected to a few RF chains by means of phase shifters, and (b) the digital part. While the former stage dramatically reduces the AP complexity and power consumption, the performance decreases as well. Consequently, properly designing the analog beamformer might be a mean to reduce the performance gap with respect to fully digital transceivers. To the best of our knowledge, there are two main works dealing with the construction of the analog beamformer as a function of slow fading channel parameters \cite{8678745,7919262}, which is also investigated in this paper and shown to outperform the previous references.

Once the analog part is designed, we investigate the uplink and downlink of two digital benchmarks: (\emph{i}) minimum mean squared error (MMSE) reception and (\emph{ii}) regularized zero forcing (RZF) precoding. Asymptotic approximations on the signal-to-interference-and-noise-ratio (SINR) are derived based on \cite{6172680}, and shown to be tight for finite-dimension systems under the previous decoding/precoding. For a given hybrid structure, and capitalizing on the asymptotic approximations, another relevant problem is studied in this paper: pilot assignment, for which a greedy algorithm based on the asymptotic expressions is provided.

Finally, we derive two novel bounds on the gap between hybrid and fully digital structures. It is shown that such bounds only depend on the channel matrix eigenvalues.

\section{System Model}
Consider a CF massive MIMO system composed by $M$ APs, each equipped with $N$ antennas and $L (\le N)$ RF chains serving $K$ single antenna users (UEs). We assume each AP is connected to a central processing unit (CPU) through high capacity fronthaul links. 
%The block fading channel model is used  for the links between APs and UEs where, during $\tau_c$ resource units the channel remains constant.
Denote by $\hmk \in \mathbb{C}^{N \times 1}$ the channel between AP $m$ and UE $k$. Then
\begin{align}
    \hmk \sim  \mathcal{N}_{\mathbb{C}}(\boldsymbol{0},\boldsymbol{R}_{m,k}),
\end{align}
with $\boldsymbol{R}_{m,k}$ being the spatial correlation matrix. Each AP performs hybrid beamforming with the aim of reducing the number of RF chains at the transceivers, and therefore their cost and complexity. Particularly, each AP contains an analog matrix $\boldsymbol{W}_m \in \mathbb{C}^{N \times L}$ such that $ \big( | \boldsymbol{W}_m | \big)_{n,l} = \frac{1}{\sqrt{N}}$, emulating phase shifters and whose entries will be designed later. As a consequence, the effective channel between AP $m$ and UE $k$ is represented by $\gmk \in \mathbb{C}^{L \times 1}$ 
\begin{align}
    \gmk = \boldsymbol{W}_m^{*} \hmk.
\end{align}
Hence, $\gmk \sim \mathcal{N}_{\mathbb{C}}(\boldsymbol{0},\boldsymbol{R}_{m,k}^{(g)})$ with $\boldsymbol{R}_{m,k}^{(g)} = \boldsymbol{W}_m^* \boldsymbol{R}_{m,k} \boldsymbol{W}_m$. %Provided that $L\leq N$, the number of RF chains needed at each AP is less than if a fully digital architecture was implemented.

\subsection{Channel Estimation Process}
A portion of the total number of resource units, the latter denoted by $\tau_c$, is used for channel estimation. During $\tau (\leq \tau_c)$ channel uses, UE $k$ is assigned a pilot $\boldsymbol{\phi}_k \in \mathbb{C}^{\tau \times 1}$ with $||\boldsymbol{\phi}_k||^2 = \tau$ and
the pilot matrix is denoted by $\boldsymbol{\Phi} = (\boldsymbol{\phi}_1,\dots,\boldsymbol{\phi}_K) \in \mathbb{C}^{\tau \times K}$. Upon pilot transmission at a certain power $p^{(t)}$, the observations at the $m$th AP are
\begin{align}
    \boldsymbol{Y}_m = \sqrt{p^{(t)}} ( \boldsymbol{g}_{m,1},\dots,\boldsymbol{g}_{m,K} ) \boldsymbol{\Phi}^{\textrm{{T}}} + \boldsymbol{W}_m^*\boldsymbol{Z}_m,
\end{align}
with $\boldsymbol{Z}_m  \sim  \mathcal{N}_{\mathbb{C}}(\boldsymbol{0},\sigma^2\boldsymbol{I}_{N})$ for $\sigma^2$ being the noise power.  Standard MMSE estimation leads to the next estimates \cite{EstTheory}
\begin{align}
    \gmkh = \sqrt{p^{(t)}} \boldsymbol{R}_{m,k}^{(g)}( \boldsymbol{\phi}_k \otimes \boldsymbol{I}_L )^* \boldsymbol{\Psi}_m^{-1} \text{vec}(\boldsymbol{Y}_m) ,
\end{align}
with 
\begin{align}
    \boldsymbol{\Psi}_m = {p^{(t)}} ( \boldsymbol{\Phi} \otimes \boldsymbol{I}_L ) \boldsymbol{R}_{m}^{(g)} ( \boldsymbol{\Phi} \otimes \boldsymbol{I}_L )^* +   \sigma^2 \boldsymbol{I}_{\tau} \otimes  \boldsymbol{W}_m^*\boldsymbol{W}_m ,
\end{align}
for $\boldsymbol{R}_{m}^{(g)} = \mathrm{diag}\{ \boldsymbol{R}_{m,k}^{(g)} \ls \mathrm{for} \ls k=1,\dots,K  \}$. It can be verified that $\gmk = \gmkh + \gmkt$ with $\gmkt$ denoting the error, uncorrelated with the estimate. More concretely, $\gmkh \sim \mathcal{N}_{\mathbb{C}}(\boldsymbol{0},\boldsymbol{\Gamma}_{m,k}^{(g)})$ with $\boldsymbol{\Gamma}_{m,k}^{(g)}$ defined by
\begin{align}
    \boldsymbol{\Gamma}_{m,k}^{(g)} & = \mathbb{E}\{ \gmkh \gmkh^* \} \\
    & = \boldsymbol{R}_{m,k}^{(g)}( \boldsymbol{\phi}_k \otimes \boldsymbol{I}_L )^* \boldsymbol{\Psi}_m^{-1} ( \boldsymbol{\phi}_k \otimes \boldsymbol{I}_L ) \boldsymbol{R}_{m,k}^{(g)},
\end{align}
and the channel error following $\gmkt \sim \mathcal{N}_{\mathbb{C}}(\boldsymbol{0},\boldsymbol{C}_{m,k}^{(g)})$ with $\boldsymbol{C}_{m,k}^{(g)} = \boldsymbol{R}_{m,k}^{(g)} - \boldsymbol{\Gamma}_{m,k}^{(g)}$.

\subsection{Scalable Cell-Free}
Although CF networks allow users to establish connectivity to multiple APs, scalability must be taken into account. Therefore only a subset of APs jointly serve a particular user. Hence, we define by $\mathcal{F}_k$ the subset of APs involved in the decoding of the $k$th UE and by $\mathcal{U}_m$ the subset of UEs treated as signal by AP $m$. Thus, the binary matrix $\boldsymbol{M} = (\boldsymbol{m}_1,\dots,\boldsymbol{m}_K ) \in \mathbb{Z}_2^{M \times K}$ whose entries are
\begin{equation}\label{eq:MsDef}
  \left( \boldsymbol{M} \right)_{m,k} =
  \begin{cases}
                                   1  & \text{if $k \ls \in \ls \mathcal{U}_m$} \\
                                  0  & \text{otherwise}
 \end{cases},
\end{equation}
accounts for scalability. Provided that each AP observes an $L$-dimensional signal after the hybrid beamforming stage, the expanded version of $\boldsymbol{M}$ is $\Ms = \boldsymbol{M} \otimes \boldsymbol{1}_L$ with $\boldsymbol{1}_L$ an $L$-dimensional vector of ones. The complementary matrix $\Mi = \boldsymbol{1} - \Ms$ accounts for the disregarded UEs per AP.

\subsection{Uplink \& Downlink Data Transmission}
After data transmission, the signal collected by the $M$ APs is $\boldsymbol{y} = ( \boldsymbol{y}_1,\dots,\boldsymbol{y}_M)^{\text{T}} \in \mathbb{C}^{ML \times 1}$ with $\boldsymbol{y}_m \in \mathbb{C}^{L \times 1}$
\begin{align}\label{eq:y}
    \boldsymbol{y} & =  (\Ms \circ \boldsymbol{G}) \x + (\Mi \circ \boldsymbol{G} ) \x + \boldsymbol{W}^*  \boldsymbol{n},
\end{align}
with $\circ$ denoting the Hadamard product,  $\boldsymbol{G} %= ( \boldsymbol{g}_1,\dots,\boldsymbol{g}_K) 
\in \mathbb{C}^{ML \times K}$ being the effective channel matrix whose entries are %$(m,k)$ entry is
$(\boldsymbol{G})_{m,k} = \boldsymbol{g}_{m,k} \in \mathbb{C}^{L \times 1}$. Vector $\x = (\sqrt{p_1}s_1,\dots,\sqrt{p_K}s_K)^{\rm T}$ for given UE transmit powers and symbols, denoted by $p_k$ and $s_k$, respectively. Finally, $\boldsymbol{W} = \textrm{diag}\{ \boldsymbol{W}_m \ls \textrm{for} \ls m=1,\dots,M\}$ and   $\boldsymbol{n} = (\boldsymbol{n}_1,\dots,\boldsymbol{n}_M)^{\mathrm{T}}$ where $\boldsymbol{n}_m \sim \mathcal{N}_{\mathbb{C}}(\boldsymbol{0},\sigma^2 \boldsymbol{I}_N)$.

In the downlink, the APs jointly precode the users data. More particularly, the precoder intended for UE $k$ is denoted by $\boldsymbol{v}_k \in \mathbb{C}^{ML\times1}$ and after data transmission, the signal collected at UE $k$ is
\begin{align}\label{eq:dlsignal}
    y_k = \sum \limits_{i=1}^K \boldsymbol{g}_k^{*}\boldsymbol{v}_i \sqrt{p_i} s_i + n_k,
\end{align}
where $n_k \sim \mathcal{N}_{\mathbb{C}}({0},\sigma^2)$. %Note that scalability in the donwlink is included in the design of the precoding vector and will be discussed in the next section.

\section{Spectral Efficiency Analysis}
\subsection{Uplink MMSE Reception}
Provided that for UE $k$ only $|\mathcal{F}_k|$  APs are relevant, taking  the  rows of $\boldsymbol{y}$ associated to $\mathcal{F}_k$ produces the following reduced signal model%, $\boldsymbol{y}_k \in \mathbb{C}^{|\mathcal{F}_k|L \times 1}$
\begin{align}\label{eq:yk}
    \boldsymbol{y}_k & = %\Ms_k \circ \boldsymbol{G}_k \x + \Mi_k \circ \boldsymbol{G}_k \x + \boldsymbol{W}^*_k  \boldsymbol{n} \\
    %& = 
    \underbrace{\Ms_k \circ \Gmh_k \x}_\text{signal}  + \underbrace{ \big( \Ms_k \circ \Gmt_k + \Mi_k \circ \Gm_k \big) \, \x +  \boldsymbol{W}^*_k  \boldsymbol{n}}_\text{ effective noise: $\boldsymbol{z}_k$ },
\end{align}
where matrices in \eqref{eq:yk} are the reduced version of the original matrices which contain the rows related to $\mathcal{F}_k$ and all columns. Moreover, $\boldsymbol{z}_k \sim \mathcal{N}_{\mathbb{C}}(\boldsymbol{0},\boldsymbol{\Sigma}_k)$ with $\boldsymbol{\Sigma}_k$ being a block diagonal matrix $\boldsymbol{\Sigma}_k =  \mathrm{diag}\{ \boldsymbol{\Sigma}_{k,m} \in \mathbb{C}^{L\times L} \ls \mathrm{for} \ls m \in \mathcal{F}_k  \} $ where the diagonal terms are
\begin{align}
    \boldsymbol{\Sigma}_{k,m} = \sum_{i \in \mathcal{U}_m} \boldsymbol{C}_{m,i}^{(g)}p_i + \sum_{i \notin \mathcal{U}_m} \boldsymbol{R}_{m,i}^{(g)}p_i + \sigma^2 \Wmh{} \Wm.
\end{align}

In the uplink, the combiner maximizing the SINR is the MMSE, achieving a maximum value of

%For a generic combiner, $\boldsymbol{v}_k \in \mathbb{C}^{|\mathcal{F}_k|L \times 1}$, User $k$ can achieve the following instantaneous SINR
%\begin{align}\label{eq:SINRk1}
%    \mathrm{SINR}_k = \frac{ |\boldsymbol{v}_k^* \boldsymbol{\hat{g}}_k|^2 p_k }{ \sum
%    \limits_{i\neq k}^{K}|\boldsymbol{v}_k^* ( \boldsymbol{m}_{k,i}^{(s)}\circ \boldsymbol{\hat{g}}_i)|^2 p_i + \boldsymbol{v}_k^* \boldsymbol{\Sigma}_k \boldsymbol{v}_k }.
%\end{align}
%However, the MMSE combiner maximizes \eqref{eq:SINRk1}, providing an SINR of:
\begin{align}\label{eq:SINRmmse}
    \mathrm{SINR}_k = \boldsymbol{\hat{g}}_k^{*} \bigg(  \sum 
    \limits_{i\neq k}^{K} ( \boldsymbol{m}_{k,i}^{(s)}\circ \boldsymbol{\hat{g}}_i) ( \boldsymbol{m}_{k,i}^{(s)}\circ \boldsymbol{\hat{g}}_i)^{*}p_i + \boldsymbol{\Sigma}_k \bigg)^{-1}  \boldsymbol{\hat{g}}_k.
\end{align}
where $\boldsymbol{\hat{g}}_k$ and $\boldsymbol{\hat{g}}_i$ are the $k$th and $i$th columns of $\boldsymbol{\hat{G}}_k$, respectively, and a similar definition applies to $\boldsymbol{m}_{k,i}^{(s)}$. As a consequence, after accounting for the pilot overhead $\frac{\tau}{\tau_c}$, the  ergodic spectral efficiency that the $k$th UE can achieve is
\begin{equation}\label{eq:SEff}
    \mathrm{SE}_k = \left(1 - \frac{\tau}{\tau_c} \right) \mathbb{E}\{ \log_2 ( 1 + \mathrm{SINR}_k )\}.
\end{equation}

\subsection{Downlink RZF Precoding}
Various precoding strategies can be used to encode the users data. However,  RZF   provides an outstanding performance as  studied in the literature. More particularly, the subset RZF precoding, denoted by $\boldsymbol{V} = (\boldsymbol{v}_1,\dots,\boldsymbol{v}_K)$, follows
\begin{align}\label{eq:RZFcomb}
    \boldsymbol{V} & = (\boldsymbol{v}_1,\dots,\boldsymbol{v}_K) \\%(\Ms \circ \boldsymbol{\hat{G}}) \boldsymbol{\Lambda} \big[ (\Ms \circ \boldsymbol{\hat{G}})^*(\Ms \circ \boldsymbol{\hat{G}}) + \rho \boldsymbol{I} \big]^{-1} \\
    %& = 
    & = \big[ (\Ms \circ \boldsymbol{\hat{G}})(\Ms \circ \boldsymbol{\hat{G}})^* + \rho \boldsymbol{I}_{ML} \big]^{-1} (\Ms \circ \boldsymbol{\hat{G}})\boldsymbol{\Lambda} .
\end{align}

with $\rho$ being the regularitzation parameter and $\boldsymbol{\Lambda} = \mathrm{diag}( \lambda_1,\dots,\lambda_K)$. Different formulations can be used for $\lambda_k$, such as to ensure (\emph{i}) $\mathbb{E} \{  ||\boldsymbol{W}\boldsymbol{v}_k||^2 \} \leq 1$ or (\emph{ii}) $||\boldsymbol{W}\boldsymbol{v}_k||^2 \leq  1$. In our case, since perfect CSI is not available, we use the former formulation. Once User $k$ receives $y_k$, as defined in Eq. \eqref{eq:dlsignal}, the following spectral efficiency can be achieved:
\begin{align}
    \mathrm{SE}_k = \left(1 - \frac{\tau}{\tau_c} \right)  \log_2 ( 1 + \mathrm{SINR}_k ),
\end{align}
with 
\begin{align}\label{eq:SINRRZF}
    \mathrm{SINR}_k = \frac{|\mathbb{E}\{\boldsymbol{g}_k^{*} \boldsymbol{v}_k   \}  |^2 p_k}{ \sum \limits_{i\neq 1}^K \mathbb{E}\{|\boldsymbol{g}_k^{*} \boldsymbol{v}_i|^2\}p_i + \mathrm{var}( \boldsymbol{g}_k^{*} \boldsymbol{v}_k )p_k + \sigma^2}
\end{align}

\section{Asymptotic Analysis}
To evaluate the previous SINR expressions, we consider the asymptotic regime, $MN,K\to\infty$ with finite $MN/K$ and investigate the convergence of the spectral efficiency expressions  to deterministic limits. Provided that the subsets account for the non-zero entries in the random matrices, it is required that they grow with the network as well, i.e., $|\mathcal{F}_k|N,|\mathcal{U}_m|\to \infty$ $\forall \ls k,m$. The premises for this convergence need the involved matrices to satisfy two technical conditions:  (a) the inverse of the resolvent matrix in \eqref{eq:SINRmmse} and \eqref{eq:RZFcomb} to exist,  ensured by  $\boldsymbol{\Sigma}_k$ and $\rho\boldsymbol{I}_{ML}$, respectively, and that (b) $\boldsymbol{\Gamma}_k^{(g)} = \text{diag}\{{m}_{m,k} \cdot \boldsymbol{\Gamma}_{m,k}^{(g)} \ls m=1,\dots,M \}$ has uniformly bounded spectral norm, for ${m}_{m,k}$ being the $(m,k)$ element of  \eqref{eq:MsDef}. Under these conditions, the following approximations can be made.
%\begin{align}
%    \boldsymbol{\Gamma}_k^{(g)} & = \mathbb{E} \Big\{  \big( \boldsymbol{m}^{\text{(s)}}_k \circ \boldsymbol{\hat{g}}_{k}  \big )\big( \boldsymbol{m}^{\text{(s)}}_k \circ  \boldsymbol{\hat{g}}_{k}   \big )^{*} \Big\}  \\
%     & = \mathrm{diag} \big \{  m_{k,m}^{(s)} \boldsymbol{\Gamma}_{m,k}^{(g)}\; \forall m \big \},
%\end{align}

\begin{theorem}\label{Th:T1}
For $|\mathcal{F}_k|N,|\mathcal{U}_m|\to \infty$ $\forall \ls k,m$ and UL MMSE combining, $\mathrm{SINR}_k \approx \overline{\mathrm{SINR}}_k$ with $\overline{\mathrm{SINR}}_k$ given in \eqref{eq:SINRAppr}.
\begin{align}\label{eq:SINRAppr}
    \overline{\mathrm{SINR}}_k =\frac{p_k}{|\mathcal{F}_k|N}  \sum \limits_{m \in \mathcal{F}_k} \mathrm{tr} \Big[ \boldsymbol{\Gamma}_{m,k}^{(g)} \boldsymbol{T}_{m,k} \Big],
\end{align}
where
\begin{align}\label{eq:Tmk}
    \boldsymbol{T}_{m,k} = \bigg( \frac{1}{|\mathcal{F}_k|N} \sum \limits_{i = 1}^K \frac{{m}_{m,k} \cdot  \boldsymbol{\Gamma}_{m,i}^{(g)} }{1 + e_{i}}p_i + \frac{1}{|\mathcal{F}_k|N}\boldsymbol{\Sigma}_{m,k} \bigg)^{-1}.
\end{align}
The coefficients $e_{i}$ are obtained iteratively, $e_{i} = \lim_{n \to \infty} e_{i}^{(n)}$, given $e_{i}^{(0)} = |\mathcal{F}_i|N$ and the recursion in \eqref{eq:coeffek}.
    
\begin{align}\label{eq:coeffek}
    e_{i}^{(n)} &  = p_i\mathrm{tr} \Bigg[  \boldsymbol{\Gamma}_{i}^{(g)} \bigg( \sum \limits_{j = 1}^K \frac{  \boldsymbol{\Gamma}_{j}^{(g)} p_j }{1 + e_{j}^{(n-1)}} +  \boldsymbol{\Sigma}_{i} \bigg)^{-1} \Bigg] 
\end{align}
\end{theorem}

\begin{proof}
The proof can be found in App. \ref{proof:SINRAssym}.
\end{proof}

\begin{theorem}\label{Th:T2}
For $|\mathcal{F}_k|N,|\mathcal{U}_m|\to \infty$ $\forall \ls k,m$ and DL RZF precoding, $\mathrm{SINR}_k \approx \overline{\mathrm{SINR}}_k$ with $\overline{\mathrm{SINR}}_k$ given in \eqref{eq:SINRAppr2}
\begin{align}\label{eq:SINRAppr2}
    \overline{\mathrm{SINR}}_k =  \frac{\frac{\mu_k^2}{\delta_k} p_k}{ \sum \limits_{i\neq 1}^K \frac{\theta_{k,i}}{\delta_i}p_i + \sigma^2},
\end{align}
where
\begin{align}
    \mu_k = \frac{1}{MN} \mathrm{tr}\big[ \boldsymbol{\Gamma}_k^{(g)} \boldsymbol{T} \big],
\end{align}
\begin{align}
    \delta_k = \frac{1}{(MN)^2}  \mathrm{tr}\big[ \boldsymbol{\Gamma}_k^{(g)} \boldsymbol{T}^{'}(\frac{\rho}{MN},\boldsymbol{W}^* \boldsymbol{W}  ) \big],
\end{align}
\begin{align}
    \theta_{k,i} & =  \frac{1}{(MN)^2}  \mathrm{tr}\big[ \boldsymbol{R}_k^{(g)} \boldsymbol{T}^{'}(\frac{\rho}{MN},\boldsymbol{\Gamma}_i^{(g)}) \big] + \nonumber \\ & \mspace{22mu}\frac{1}{MN} \frac{\mu_k^2 \frac{1}{MN} \mathrm{tr}\big[ \boldsymbol{\Gamma}_k^{(g)} \boldsymbol{T}^{'}(\frac{\rho}{MN},\boldsymbol{\Gamma}_i^{(g)}) \big] }{(1 + \mu_k)^2} - \nonumber \\ & \mspace{22mu} \frac{2}{MN}\mathbb{R} \bigg\{ \frac{\mu_k \frac{1}{MN} \mathrm{tr}\big[ \boldsymbol{\Gamma}_k^{(g)} \boldsymbol{T}^{'}(\frac{\rho}{MN},\boldsymbol{\Gamma}_i^{(g)})  }{1 +  \mu_k} \bigg\},
\end{align}

%for given
\begin{align}
    \boldsymbol{T} = \bigg( \frac{1}{MN}\sum \limits_{i = 1}^K \frac{ \boldsymbol{\Gamma}_{i}^{(g)} }{1 + e_{i}} + \frac{\rho}{MN} \boldsymbol{I}_{ML} \bigg)^{-1}.
\end{align}

The coefficients $e_{i}$ are obtained iteratively with $e_{i} = \lim_{n \to \infty} e_{i}^{(n)}$, given $e_{i}^{(0)} = MN$ and the recursion in \eqref{eq:coeffek2}
\begin{align}\label{eq:coeffek2}
    e_{k}^{(n)} = \mathrm{tr} \Bigg[  \boldsymbol{\Gamma}_{k}^{(g)} \bigg(  \sum \limits_{i = 1}^K \frac{ \boldsymbol{\Gamma}_{i}^{(g)} }{1 + e_{i}^{(n-1)}} + \rho \boldsymbol{I}_{ML} \bigg)^{-1} \Bigg].
\end{align}
Moreover, matrix 
\begin{align}
    \boldsymbol{T}^{'}(\frac{\rho}{MN},\boldsymbol{\Gamma}_{i}^{(g)}) = \Tmat \boldsymbol{\Gamma}_{i}^{(g)} \Tmat + \Tmat \frac{1}{M}\sum \limits_{k=1}^K \frac{\boldsymbol{\Gamma}_{k}^{(g)} e_k^{'}}{(1+e_k)^2} \Tmat,
\end{align}
and coefficients $\boldsymbol{e}^{'}(\frac{\rho}{MN}) = ({e}_1^{'},\dots,{e}_K^{'})$ are calculated as
\begin{align}
    \boldsymbol{e}^{'}(\frac{\rho}{MN}) = \big( \boldsymbol{I}_K - \boldsymbol{J} \big)^{-1}\boldsymbol{v}(\frac{\rho}{MN}),
\end{align}
with $\boldsymbol{J} \in \mathbb{C}^{K \times K}$ and $\boldsymbol{v}(\frac{\rho}{MN}) \in \mathbb{C}^{K \times 1}$ defined as
\begin{align}
    \big( \boldsymbol{J} \big)_{k,l} =  \frac{\frac{1}{MN} \mathrm{tr} \big[ \boldsymbol{\Gamma}_{k}^{(g)} \Tmat \boldsymbol{\Gamma}_{l}^{(g)} \Tmat \big] }{MN(1 + e_l)^2},
\end{align}
and
\begin{align}
    \big( \boldsymbol{v}(\frac{\rho}{MN}) \big)_{k} = \frac{1}{MN} \mathrm{tr} \big[ \boldsymbol{\Gamma}_{k}^{(g)} \Tmat  \boldsymbol{\Gamma}_{i}^{(g)} \Tmat \big].
\end{align}
\end{theorem}

\begin{proof}
The proof can be found in App. \ref{proof:SINRAssymDL}.
\end{proof}

From the continuous mapping theorem \cite{Mann1943OnSL}, the following holds: $\mathrm{SE}_k \approx \left(1 - \frac{\tau}{\tau_c} \right)   \log_2 ( 1 + \overline{\mathrm{SINR}}_k )$ with the corresponding $\overline{\mathrm{SINR}}_k$ provided above for $MN$ and $K$  $\to \infty$.

\section{Spectral Efficiency Optimization}
Note that the asymptotic SE approximations derived in the previous section only depend on large scale parameters. Therefore,  we can formulate different asymptotic optimization problems. However, with the aim of increasing fairness in the network we focus on the following max-min problem:
\begin{equation}\label{opt:Problem}
    \begin{aligned}
        \max_{\boldsymbol{W} , \ls  \boldsymbol{\Phi}} & \min_k  \overline{\mathrm{SINR}}_k.\\
        %\textrm{s.t.} \quad & p_k \leq p_{max}\\
        \textrm{s.t.} \quad & \big(|\boldsymbol{W}_m|\big)_{n,l} = \frac{1}{\sqrt{N}}
    \end{aligned}
\end{equation}
where the optimization variables are two: (\emph{i}) analog beamforing matrix $\boldsymbol{W}$ and (\emph{ii}) pilot matrix, studied separately.

\subsection{Analog Beamformer Design}\label{Sec:BFMat}
The design of $\boldsymbol{W} = \mathrm{diag}\{ \boldsymbol{W}_m \ls \mathrm{for} \ls m=1,\dots,M  \}$ is challenging given the complexity of the SINR. %Additionally, $\boldsymbol{W}_m$ should  only depend on large scale parameters, i.e. the channel covariances, and satisfy that its entries are roots of unity. 
Therefore, directly solving \eqref{opt:Problem} poses a major challenge. However, under perfect CSI, some algebraic properties on $\boldsymbol{W}_m$ can be extracted and therefore used for its design. Concretely, we first disregard the unit-modulus constraint and after SVD decomposition $\boldsymbol{W}$ factorizes as $\boldsymbol{W} = \boldsymbol{U} \boldsymbol{Q}$ with semi-unitary $\boldsymbol{U}$, i.e. $\boldsymbol{U}^* \boldsymbol{U} = \boldsymbol{I}_{ML}$.

\begin{prop}\label{prop:WMMSE}
Under perfect CSI UL-MMSE reception, any nonsingular $\boldsymbol{Q}$ provides maximum SINR.
\end{prop}
\begin{proof}
The proof can be found in App. \ref{proof:WMMSE}.
\end{proof}
According to \cite{6200372},  $\sum \limits_{i\neq k}^K |\boldsymbol{g}_k^{*} \boldsymbol{v}_i|^2p_i  + \sigma^2 \approx \sum \limits_{i\neq k}^K |\boldsymbol{g}_i^{*} \boldsymbol{v}_k|^2p_i  + \sigma^2$. Under the condition that the previous approximation is tight, the following proposition, which is similar to the result obtained in \cite{7919262} for another metric, can be obtained.
\begin{prop}\label{prop:WRZF}
Under perfect CSI DL-RZF precoding, the SINR is maximum when $\boldsymbol{Q}$ is semi-unitary: $\boldsymbol{Q} \boldsymbol{Q}^* = \boldsymbol{I}_{ML}$. 
\end{prop}
\begin{proof}
The proof can be found in App. \ref{proof:WRZF}.
\end{proof}

%Ideally, the system wants the same $\boldsymbol{W}_m$ for uplink and downlink.
In order to full-fill both propositions, for UL and DL, $\boldsymbol{Q}$ can be set to %a semi-unitary matrix and focus on the design of $\boldsymbol{U}$. More particularly, a special semi-unitary matrix is
$\boldsymbol{Q} = \boldsymbol{I}_{ML}$ and therefore $\boldsymbol{W} = \boldsymbol{U}$ meaning that the analog matrix should have orthogonal columns. The idea behind having   orthogonal columns is that interference is reduced.  %in the uplink and downlink, respectively. 
To the best of our knowledge, there are two  ways of smartly creating $\boldsymbol{W}$ explained in \cite{8678745} and \cite{7919262}, respectively. While the latter is based on perfectly known channels, the former fails to capture the complete spectrum of the channel covariance matrices. In this work, we propose a method that takes into account all possible eigenvectors/eigenvalues of all $\boldsymbol{R}_{m,k}$ with the aim of maximizing the minimum average UE power signal, which is shown to maximize the minimum SINR in our simulations. More particularly, $\boldsymbol{R}_{m,k} = \boldsymbol{V}_{m,k} \boldsymbol{\Lambda}_{m,k} \boldsymbol{V}_{m,k}^*$ with $\boldsymbol{V}_{m,k}$ having orthonormal column vectors and $\boldsymbol{\Lambda}_{m,k} = \mathrm{diag}(\lambda_{m,k}^{(1)},\dots,\lambda_{m,k}^{(N)} )$ containing the $N$ eigenvalues of $\boldsymbol{R}_{m,k}$. Note that the average signal power for UE $k$ is given by $\sum \limits_{m \in \mathcal{F}_k} \mathrm{tr}( \boldsymbol{W}_m^*\boldsymbol{R}_{m,k} \boldsymbol{W}_m)$. Therefore, such a expression is maximized whenever the columns of $\boldsymbol{W}_m$ match the eigenvectors of $\boldsymbol{R}_{m,k}$. However, not all UEs and their respective eigenmodes can be captured by $\boldsymbol{W}_m$. A selection of  $L$ out of  $NK$ should be made. As a consequence, we define the UE average signal power as
\begin{align}
    S_k = \sum \limits_{m \in \mathcal{F}_k} \sum \limits_{n=1}^N \alpha_{m,k}^{(n)} \lambda_{m,k}^{(n)}.
\end{align}
where $\alpha_{m,k}^{(n)}$ is a binary optimization variable scheduling the eigenvectors to the columns of $\boldsymbol{W}_m$. Therefore, the following optimization problem with respect to $\alpha_{m,k}^{(n)}$ can be formulated:

\begin{equation}\label{opt:ProblemAlpha}
    \begin{aligned}
        \max_{\alpha_{m,k}^{(n)}} & \min_k  S_k\\
        \textrm{s.t.} \quad & \alpha_{m,k}^{(n)} \in \{0,1\}\\
        & \sum \limits_{k=1}^{K}\sum \limits_{n=1}^N \alpha_{m,k}^{(n)} \leq L
    \end{aligned}
\end{equation}
The reverse-delete algorithm is capable of efficiently solving \eqref{opt:ProblemAlpha} without the need of an exhaustive search. The surviving $\alpha_{m,k}^{(n)}$ determine which eigenvectors of which users will compose the columns of $\boldsymbol{W}_m$. However, note that $ \boldsymbol{{W}}_m$ for $m=1,\dots,M$ does not necessarily have orthogonal columns given that, most likely, eigenvectors from multiple users will be used to construct the analog matrices. As a consequence, neither Prop. \ref{prop:WMMSE} nor Prop. \ref{prop:WRZF} are satisfied. Thus, the final unconstrained analog beamformers are obtained by  $\boldsymbol{W}_m^{(p)} = \mathcal{P}( \boldsymbol{{W}}_m )$ where $\mathcal{P}( \boldsymbol{A}_m )$ is the projection of matrix $\boldsymbol{A}_m$ into an orthonormal basis. 

Still, $\boldsymbol{W}_m^{(p)}$ is not only composed by phase shifters, i.e. the entries are not roots of unity. Therefore, for given $\boldsymbol{W}_m^{(p)}$, we aim at solving the following optimization problem:
\begin{equation}\label{Prob:ProjectW}
    \begin{aligned}
    & \underset{\boldsymbol{{\hat{W}}}_m}{\text{min}}
    & & || \boldsymbol{{W}}_m^{(p)} - \boldsymbol{\hat{W}}_m ||_{\text{F}}^2  \\
    & \text{s.t.} & &  | [\boldsymbol{\hat{W}}_m]_{n,l} | = \frac{1}{\sqrt{N}}
    \end{aligned}.
\end{equation}
Although the optimal solution is obtained  by taking the phase of the eigenvectors in $\boldsymbol{W}_m^{(p)}$, the orthogonality between columns achieved by $\mathcal{P}(\cdot)$ would be broken. Therefore, we modify our receiver. We add an \textit{orthogonality compensation matrix}  into our digital processing \cite{7919262}. More concretely, for a constrained analog beamformer $\boldsymbol{\hat{W}}_m$, its SVD results in $\boldsymbol{\hat{W}}_m = \boldsymbol{\hat{U}}_m \boldsymbol{\hat{D}}_m \boldsymbol{\hat{V}}_m^{*}$.
The orthogonality compensation matrix, denoted by $\boldsymbol{F}_m$, is defined as
\begin{align}\label{eq:CompMat}
    \boldsymbol{F}_m = \boldsymbol{\hat{V}}_m \boldsymbol{\hat{D}}_m^{-1} \boldsymbol{\hat{V}}_m^{ *}.
\end{align}
Therefore, adding such a compensation matrix allows us to improve the design of the analog matrix exploiting the following proposition.
\begin{prop}\label{prop:WAF}
Assume that instead of using $\boldsymbol{W}_m^{(p)}$ as the analog matrix, $\boldsymbol{W}_m^{(p)} \boldsymbol{A}_m$ is the new analog beamformer with $\boldsymbol{A}_m \in \mathbb{C}^{L \times L}$ nonsingular. The product between $\boldsymbol{W}_m^{(p)} \boldsymbol{A}_m \boldsymbol{F}_m$ provides the same optimality as $\boldsymbol{W}_m^{(p)}$ and therefore $\boldsymbol{W}_m^{(p)} \boldsymbol{A}_m$ is an optimal unconstrained analog matrix.
\end{prop}
\begin{proof}
The proof can be found in App. \ref{proof:WAF}
\end{proof}

Using the previous proposition, the initial unconstrained beamformer $\boldsymbol{W}_m^{(p)}$ can be replaced by $\boldsymbol{W}_m^{(p)} \boldsymbol{A}_m$ without a performance degradation as long as $\boldsymbol{A}_m$ is nonsingular. As a consequence, we can formulate the following optimization problem:
 \begin{equation}\label{Prob:ProjectW}
    \begin{aligned}
    & \underset{\boldsymbol{\hat{W}}_m, \boldsymbol{A}_m}{\text{min}}
    & & || \boldsymbol{\hat{W}}_m - \boldsymbol{{W}}_m^{(p)} \boldsymbol{A}_m ||_{\text{F}}^2  \\
    & \text{s.t.} & &  | [\boldsymbol{\hat{W}}_m]_{n,l} | = \frac{1}{\sqrt{N}}
    \end{aligned}
\end{equation}
Thanks to the degrees of freedom added by $\boldsymbol{A}_m$, the constrained analog beamformer $\boldsymbol{\hat{W}}_m$, can be made closer to the unconstrained one $\boldsymbol{W}_m^{(p)}$.
By alternating minimization, we split the previous problem into two sub-problems: (\emph{i}) find the optimal $\boldsymbol{A}_m$ for fixed $\boldsymbol{\hat{W}}_m$ and (\emph{ii}) find the optimal $\boldsymbol{\hat{W}}_m$ for fixed $\boldsymbol{A}_m$. The solution to the previous subproblems is
\begin{equation}
    \boldsymbol{A}_m = \boldsymbol{{W}}_m^{(p) \ls *} \boldsymbol{\hat{W}}_m ,
\end{equation}
\begin{align}
    \boldsymbol{\hat{W}}_m = \frac{1}{\sqrt{N}} \text{exp} \angle (\boldsymbol{{W}}_m^{(p)} \boldsymbol{A}_m ).
\end{align}
An iterative process based on the block coordinate descend method follows until convergence is reached \cite{BCD}. % as summarized in Alg. \ref{alg:ItAW}. 
Therefore, a constrained analog matrix will be obtained and thus from Eq. \eqref{eq:CompMat} we can create  $\boldsymbol{F}_m$ that goes into the baseband (or digital) part. As a consequence, the equivalent channel between AP $m$ and UE $k$ has an extra component:
\begin{align}
    \boldsymbol{g}_{m,k} = \boldsymbol{F}_m^{*} \boldsymbol{\hat{W}}_m^{*} \boldsymbol{h}_{m,k}.
\end{align}
%\textit{Note that adding $\boldsymbol{F}_m^{*}$ does not incur in a burden of extra operations in the digital part.}

%\begin{algorithm}
%    \caption{Generate constrained analog beamformer}\label{alg:ItAW}
%    \begin{algorithmic}
%        \Require Initial $\boldsymbol{\hat{W}}_m(t) = \frac{1}{\sqrt{N}} \text{exp} \angle (\boldsymbol{\hat{W}}_m^{(p)} )$ with $t = 0$ denoting the iteration number.
%        \While{$||\boldsymbol{\hat{W}}_m(t) - \boldsymbol{\hat{W}}_m(t-1)||_F > \epsilon$ }
%        \State Compute $\boldsymbol{A}_m(t) = \boldsymbol{\hat{W}}_m^{(p)*} \boldsymbol{\hat{W}}_m(t-1)$.
%        \State Update $\boldsymbol{\hat{W}}_m(t) = \frac{1}{\sqrt{N}} \text{exp} \angle (\boldsymbol{\hat{W}}_m^{(p)} \boldsymbol{A}_m(t) )$.
%    \EndWhile
%    \end{algorithmic}
%    \end{algorithm}

\subsection{Pilot Assignment Optimization}
The optimal solution to \eqref{opt:Problem} with respect to $\boldsymbol{\Phi}$ requires an exhaustive search over the set of possible pilot sequences. However, based on the correlation between effective channels: $\Delta_{k,i} = \text{tr} ( \boldsymbol{\Gamma}_k^{(g)} \boldsymbol{\Gamma}_i^{(g)} )$ for $ k \neq i$, an initial pilot assignment can be made, denoted by $\boldsymbol{\Phi}^{(0)}$. Particularly, a set of users is assigned the same pilot if their normalized cross-correlation, i.e. $\frac{\Delta_{k,i}}{ \text{tr} ( \boldsymbol{\Gamma}_k^{(g)} )  \text{tr} ( \boldsymbol{\Gamma}_i^{(g)} )}$, is minimized. Afterwards, the greedy algorithm proposed in Alg. \ref{alg:pilotassignment} combined with the asymptotic approximations can be used to iteratively update the UE pilot assignment in a max-min SINR sense. Additionally, by construction, Alg. \ref{alg:pilotassignment} converges provided that the cost function (\emph{i})  is non-decreasing and (\emph{ii}) is upper bounded.

\begin{algorithm}
    \caption{Greedy pilot assignment}\label{alg:pilotassignment}
    \begin{algorithmic}
        \Require Set of available pilots, $\mathcal{S} = \{s_1,\dots,s_{|\mathcal{S}|}\}$ and initial pilot assignment $\boldsymbol{\Phi}^{(0)}$ at iteration $j=0$.
        \State Define the cost function $\mu^{(0)} = \min_k \overline{\mathrm{SINR}}_k(\boldsymbol{\Phi}^{(0)})$.
        \While{${\mu^{(j+1)} - \mu^{(j)} \over \mu^{(j)} } > \epsilon $ }
                \State For each UE $u=1,\dots,K$ solve \begin{gather}
                    \phi_u^{(j+1)}  =  \arg \max_{s \in \mathcal{S}} \ls \nonumber \\ \min_k \overline{\mathrm{SINR}}_k(\phi_1^{(j+1)},\dots,\phi_{u-1}^{(j+1)}, s, \phi_{u+1}^{(j)},\dots,\phi_K^{(j)})
                \end{gather}
                % where $u$ is the inner iteration.
                %\State $\boldsymbol{\Phi}^{(j+1,u)} = (\phi_1^{(j+1)},\dots,\phi_u^{(j+1)},\phi_{u+1}^{(j)},\dots,\phi_K^{(j)} ) $
            %\EndFor
        %\State $\boldsymbol{\Phi}^{(j+1)} = \boldsymbol{\Phi}^{(j+1,K)}$.
        \State Update cost function $\mu^{(j+1)} = \min_k \overline{\mathrm{SINR}}_k(\boldsymbol{\Phi}^{(j+1)})$
    \EndWhile
    \end{algorithmic}
\end{algorithm}

%\subsection{Power Optimization}
%The system performance can be further improved by adjusting the transmit power of the devices. In fact, the optimal solution to the max-min SINR problem presented in \eqref{opt:Problem} can be obtained as per \cite{7583660}. More particularly, both the cost function and the constraints satisfy the definition of \textit{Competitive Utility Function} and \textit{Monotonic Constraints} (a similar proof is provided in \cite[App. C]{ICC22}). Consequently, by \cite[Alg. 1]{7583660} converges to the optimal power allocation.

%\begin{algorithm}
%\caption{GU Power Allocation Algorithm}\label{alg:GUPower}
%\begin{algorithmic}
%\Require Initial values at $t=0$ for $p_k(t) \geq 0 \mspace{6mu} \forall k$ 
%\While{$\sum_k \big( \frac{p_k(t+1) - p_k(t)}{p_k(t)} \big) > \epsilon $}
%    \State $p_k(t+1) \gets  \frac{p_k(t)}{\overline{\mathrm{SINR}}_k(t)}$.
%    \State $p_k(t+1) \gets \frac{p_k(t+1)}{\max_i p_i(t+1) } \cdot p_{\bf max}$.
%    \State Update $\overline{\mathrm{SINR}}_k(t+1)$ with the values of $p_k(t+1)$.
%\EndWhile
%\end{algorithmic}
%\end{algorithm}

\section{$M \to \infty$ Regime}
Finally, we focus on the case where $M\to\infty$. % at a faster peace than $K$, corresponding to the case where the number of APs is way larger than that of UEs. 
For simplicity, assume $\Ms = \boldsymbol{1}$ and recall that a full digital structure is the one providing the best performance in terms of SE, attained when $L=N$ and $\boldsymbol{W}_m = \boldsymbol{I}_{N}$. Then, the following can be derived. % However, an upper and lower bound on the gap between full digital and hybrid can be derived for the unconstrained analog matrix as shown in the next proposition.
\begin{prop}\label{prop:UBLBHyb}
Define the gap as the difference in SINR between full digital and hybrid. Then, there exist  lower and upper bounds for the gap,  denoted by $\delta_{\mathrm{LB}}$ and $\delta_{\mathrm{UB}}$, given by
\begin{align}
    \delta_{\mathrm{LB}} = \frac{p_k}{\sigma^2}  \sum \limits_{m =1}^M \sum \limits_{n=L+1}^N \lambda_{m,k}^{(n)}.
\end{align}

\begin{align}
    \delta_{\mathrm{UB}} = \frac{p_k}{\sigma^2}  \sum \limits_{m =1}^M \bigg( \sum \limits_{n=1}^N ( \lambda_{m,k}^{(n)} - \lambda_{m,k}^{(N-L + n)} ) +  \sum \limits_{n=L+1}^M \lambda_{m,k}^{(n)} \bigg)
\end{align}
\end{prop}
\begin{proof}
The proof can be found in App. \ref{proof:UBLBHyb}
\end{proof}
Note that if the channel matrices are rank-deficient, i.e. $\mathrm{rank}( \boldsymbol{R}_{m,k}) \leq L$, the gap can be as small as zero and therefore a hybrid structure would achieve the same performance as digital.

\section{Simulation Results}
For the purpose of performance evaluation, we consider a $200 \times 200$ $m^2$ wrapped around universe. To generate the channel model, we assume that the APs are deployed in urban environments at around 10 m, matching with the 3GPP Urban Microcell model in \cite[Table B.1.2.1-1]{3GPPCh} at an operating frequency of 2 GHz. The shadowing terms given an AP to different UEs present a certain correlation, given by the model in  \cite[Table B.1.2.2.1-4]{3GPPCh}. The number of total channel uses is $\tau_c = 200$. Unless otherwise specified, in order to take into account the effects of pilot contamination $\tau = 8$ orthogonal pilots and $K = 16$ UEs (i.e. reuse factor of two). Additionally,  each AP has $N=32$ antennas. The UE transmit power is set to $200$ mW, $\sigma^2 = -96 $ dBm and $\rho = 10^{-4}$. Moreover, to account for scalability, the  $[m,k]$ entry of $\boldsymbol{M}_{[m,k]}$ is 1 if $d_{m,k} \leq R_{\text{max}}$ for $R_{\text{max}}=90$ m, which ensures connectivity to multiple FBSs per GU for $d_{m,k}$ the Euclidean distance between AP $m$ and UE $k$. Finally, $\epsilon = 0.001$ to ensure enough iterations until convergence is reached.

The applicability of Theorems \ref{Th:T1} and \ref{Th:T2} to finite-dimensional systems is first verified in Figs. \ref{fig:ULMMSE} and \ref{fig:DLRZF}, where the approximations are denoted by RMT in the legend.  For different network setups, corresponding to $M=4$, $N=32$, $L=16$ and $M=12$, $N=32$, $L=8$, the approximations obtained in Th. \ref{Th:T1} and \ref{Th:T2} respectively are indeed accurate for $K=16$ and $\tau=8$ orthogonal pilots.

\begin{figure}%[!htb]
    \centering\includegraphics[scale=0.6]{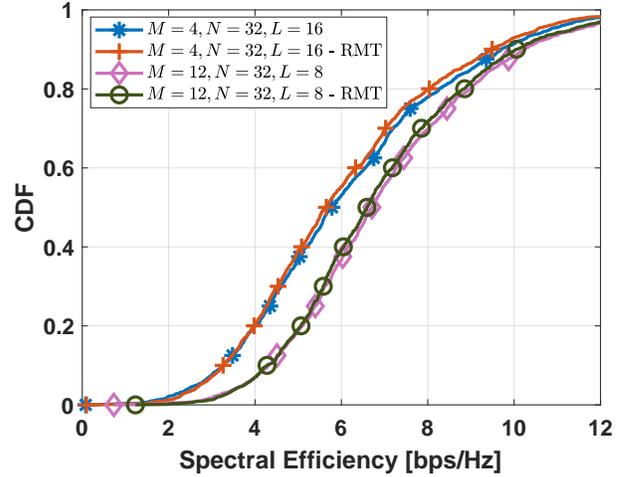}
    \caption{Exact SE vs $\left(1 - {\tau}/{\tau_c} \right)   \log_2 ( 1 + \overline{\mathrm{SINR}}_k )$ with $\overline{\mathrm{SINR}}_k$ given in Th. \ref{Th:T1}.}
    \label{fig:ULMMSE}
\end{figure}

\begin{figure}%[!htb]
    \centering\includegraphics[scale=0.6]{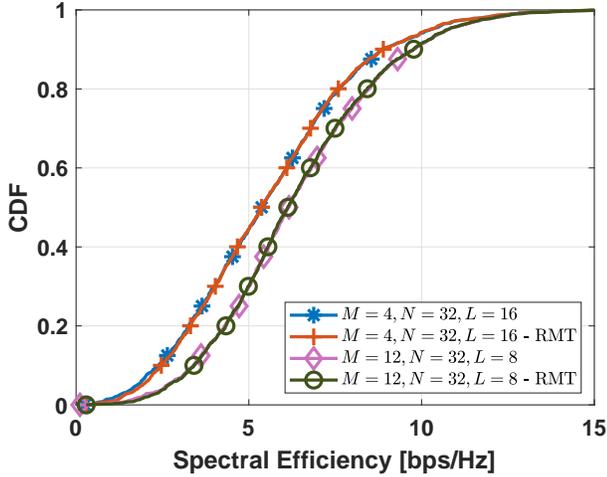}
    \caption{Exact SE vs $\left(1 - {\tau}/{\tau_c} \right)   \log_2 ( 1 + \overline{\mathrm{SINR}}_k )$ with $\overline{\mathrm{SINR}}_k$ given in Th. \ref{Th:T2}.}
    \label{fig:DLRZF}
\end{figure}

In Fig. \ref{fig:Power}, we compare the UL pilot assignment obtained by Alg. \ref{alg:pilotassignment} (Greedy) and a random assignment (RA) for different values of $N$ and $L$. For $N=L=16$ we assume a digital structure while for $N=32$ and $L=8$ the analog matrices $\boldsymbol{\hat{W}}_m$ are obtained as described in Section \ref{Sec:BFMat}. There is a visible improvement after running the greedy algorithm when the set of available pilots $\mathcal{S}$ is composed by orthogonal pilots. Additionally, the improvement in terms of minimum SE is measured and is of about 60\% and 90\% for $N=L=16$ and $N=32$, $L=8$, respectively. Similar results are obtained in the DL.

\begin{figure}%[!htb]
    \centering\includegraphics[scale=0.6]{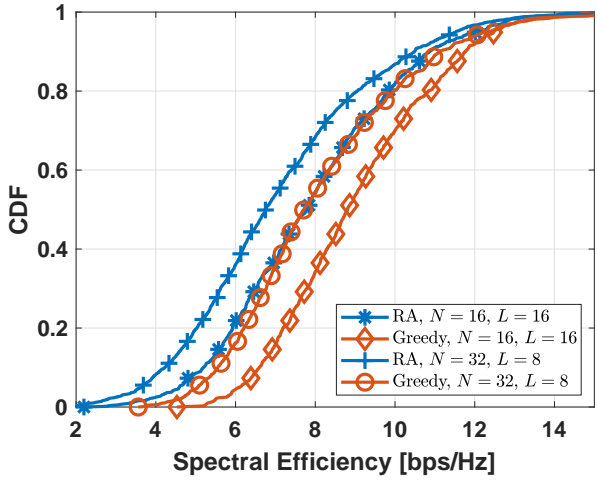}
    \caption{SE for the Greedy and RA pilot assignment schemes.}
    \label{fig:Power}
\end{figure}

Next, we analyze the performance of our hybrid beamforming method compared to the two existing techniques, called SVD \cite{8678745} and SLNR \cite{7919262}. We measure the 95\% outage SE which is a key metric in wireless systems for both the UL and DL in Figs. \ref{fig:ULW} and \ref{fig:DLW}. Clearly, our method outperforms both works in the two links, i.e. UL and DL, with gains in the range of 1-8\% and 10-35\% in the UL and DL, respectively.

\begin{figure}%[!htb]
    \centering\includegraphics[scale=0.6]{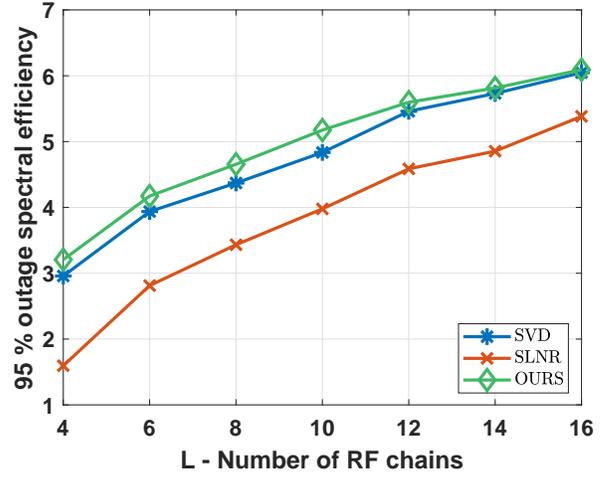}
    \caption{95\% outage UL-SE for different analog methods as a function of $L$ for $M=12$ and $N=32$.}
    \label{fig:ULW}
\end{figure}

\begin{figure}%[!htb]
    \centering\includegraphics[scale=0.6]{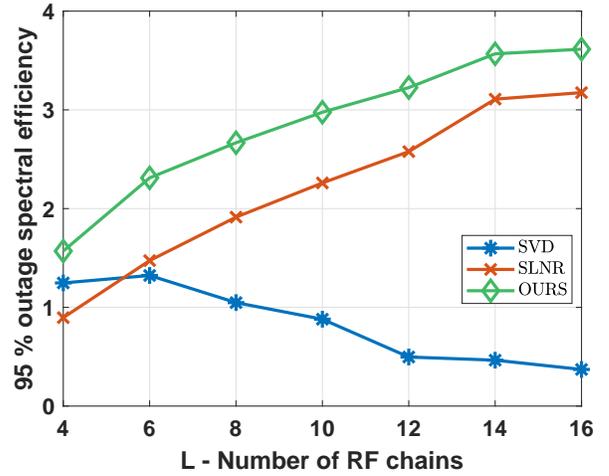}
    \caption{95\% outage DL-SE for different analog methods as a function of $L$ for $M=12$ and $N=32$.}
    \label{fig:DLW}
\end{figure}

\section{Conclusions}
This paper has investigated the use of hybrid transceivers in CF MIMO setups. After deriving asymptotic approximations for both UL and DL, we focused on solving two problems: (\emph{i}) analog beamformer and (\emph{ii}) pilot assignment. The solution to the first one is shown to outperform state-of-the-art techniques while the greedy pilot assignment highly outperforms a RA. Finally, theoretical bounds for the gap between full digital and hybrid structures are presented, showing that such a gap is highly dependant on the eigenvalues of the channel correlation matrices.

\bibliography{references}
\bibliographystyle{ieeetr}

\appendices
\section{}
\begin{theorem}\label{th:RMT1}
(\cite[Theorem 1]{6172680}) Let $\boldsymbol{{D}} \in \mathbb{C}^{M \times M}$ and $\boldsymbol{{S}} \in \mathbb{C}^{M \times M}$ be Hermitian nonnegative-definite while  $\boldsymbol{{H}} \in \mathbb{C}^{M \times K}$ is a random matrix with zero-mean independent column vectors, $\boldsymbol{h}_k$, each with covariance matrix $\frac{1}{M}\boldsymbol{\mathrm{R}}_k$. Finally, $\boldsymbol{{D}}$  and $\boldsymbol{{R}}_k$ have uniformly bounded spectral norm w.r.t. $M$. For $z>0$ and $M,K \to \infty$,
\begin{equation*}
    \frac{1}{M} \, \mathrm{tr} \! \left[ \boldsymbol{{D}} \big( \boldsymbol{{H}}\boldsymbol{{H}}^{*} + \boldsymbol{{S}} + z\boldsymbol{{I}}_M)^{-1} \right] - \frac{1}{M} \, \mathrm{tr}[ \boldsymbol{{D}} \Tmat] \stackrel{\text{a.s.}}{\to} 0 ,
\end{equation*}
where %$\Tmat\in \mathbb{C}^{M \times M}$ is defined as:
\begin{equation}\label{eq:Tmat}
    \Tmat = \bigg( \frac{1}{M} \sum \limits_{j=1}^K \frac{\boldsymbol{{R}}_j}{1+e_{j}}  + \boldsymbol{{S}} + z\boldsymbol{{I}}_M \bigg)^{\!-1}
\end{equation}
with coefficients $e_k = \text{lim}_{n\xrightarrow{}\infty} e_k^{(n)}$ for
\begin{equation}
    e_k^{(n)} =  \frac{1}{M} \, \mathrm{tr} \! \left[ \boldsymbol{{R}}_k \bigg( \frac{1}{M} \sum \limits_{j=1}^K \frac{\boldsymbol{{R}}_j}{1+e_{j}^{(n-1)}}  + \boldsymbol{{S}} + z\boldsymbol{{I}}_M \bigg)^{\!-1} \right]
\end{equation}
with initial values $e_k^{(0)}=M$.
\end{theorem}

%% RMT Th 2
\section{}
\begin{theorem}\label{th:RMT2}
(\cite[Theorem 2]{6172680}) Let $\boldsymbol{\Phi} \in \mathbb{C}^{M \times M}$ be Hermitian nonnegative-definite. Under the same conditions as Th. \ref{th:RMT1}, for $M,K \to \infty$,
\begin{align*}
    \frac{1}{M} \, \mathrm{tr} \! \left[ \boldsymbol{{D}} \big( \boldsymbol{{H}}\boldsymbol{{H}}^{*} + \boldsymbol{{S}} + z\boldsymbol{{I}}_M)^{-1} \boldsymbol{\Phi} \big( \boldsymbol{{H}}\boldsymbol{{H}}^{*} + \boldsymbol{{S}} + z\boldsymbol{{I}}_M)^{-1} \right] - \nonumber \\ \frac{1}{M} \, \mathrm{tr}[ \boldsymbol{{D}} \boldsymbol{T}^{'}(z,\boldsymbol{\Phi})] \stackrel{\text{a.s.}}{\to} 0 ,
\end{align*}
where $\boldsymbol{T}^{'}(z,\boldsymbol{\Phi})$ is defined as
\begin{equation}\label{eq:matTprime}
    \boldsymbol{T}^{'}(z,\boldsymbol{\Phi}) = \Tmat \boldsymbol{\Phi} \Tmat + \Tmat \frac{1}{M}\sum \limits_{k=1}^K \frac{\boldsymbol{R}_k e_k^{'}(z,\boldsymbol{\Phi})}{(1+e_k)^2} \Tmat
\end{equation}
with  $\Tmat$ and $e_k$ given in Th. \ref{th:RMT1} for particular $z$ and $\boldsymbol{e}^{'}(z,\boldsymbol{\Phi}) = \big(e_1^{'}(z),\dots,e_K^{'}(z)\big)$ calculated as
\begin{align}
    \boldsymbol{e}^{'}(z,\boldsymbol{\Phi}) = \big( \boldsymbol{I} - \boldsymbol{J}(z) \big)^{-1}\boldsymbol{v}(z,\boldsymbol{\Phi})
\end{align}
with $\boldsymbol{J}(z) \in \mathbb{C}^{K \times K}$ and $\boldsymbol{v}(z) \in \mathbb{C}^{K \times 1}$ defined as
\begin{align}
    \big( \boldsymbol{J}(z) \big)_{k,l} =  \frac{\frac{1}{M} \mathrm{tr} \big[ \boldsymbol{R}_k \Tmat \boldsymbol{R}_l \Tmat \big] }{M(1 + e_l)^2}
\end{align}
and
\begin{align}
    \big( \boldsymbol{v}(z,\boldsymbol{\Phi}) \big)_{k} = \frac{1}{M} \mathrm{tr} \big[ \boldsymbol{R}_k \Tmat  \boldsymbol{\Phi} \Tmat \big]
\end{align}
    
\end{theorem}

\section{Proof of Th. \ref{Th:T1}}\label{proof:SINRAssym}
Let us define matrices  $\boldsymbol{P} = \text{diag}\{ p_1,\dots,p_K \}$, 
\begin{align}
    \boldsymbol{\Omega} =  |\mathcal{F}_k|N \left( \big( \boldsymbol{{M}}^{(s)} \circ  \Gmh_k \big) \Pmat{} \big( \Ms{} \circ \Gmh_k  \big)^*  + \boldsymbol{\Sigma}_k \right)^{\!-1},
\end{align}
\begin{align}
    \boldsymbol{\Omega}_k & = \bigg( \big( \boldsymbol{{M}}^{(s)} \circ  \Gmh_k \big) \Pmat{} \big( \Ms{} \circ \Gmh_k  \big)^* - \nonumber \\&  \mspace{22mu} \big( \boldsymbol{m}_k^{(s)} \circ  \boldsymbol{\hat{g}}_k \big) \big(\boldsymbol{m}_k^{(s)} \circ  \boldsymbol{\hat{g}}_k \big)^*p_k + \boldsymbol{\Sigma}_k \bigg)^{\!-1},
\end{align} 
and $\boldsymbol{\Omega}_k^{'} = |\mathcal{F}_k|N \boldsymbol{\Omega}_k$. Then, \eqref{eq:SINRmmse} can be written as
\begin{align}
    \mathrm{SINR}_k & =  \boldsymbol{\hat{g}}_k  ^* \boldsymbol{\Omega}_k  \boldsymbol{\hat{g}}_k   \, p_k \\
    & = \frac{p_k}{|\mathcal{F}_k|N} \, \mathrm{tr} \! \left[  \boldsymbol{\hat{g}}_k   \boldsymbol{\hat{g}}_k^* \boldsymbol{\Omega}_k^{'} \right].
\end{align}
For $|\mathcal{F}_k|N$,$|\mathcal{U}_m|$ $\xrightarrow{}\infty$ $\forall \ls k,m$, we have
\begin{align}\label{eq:LargeMKSINR}
    \frac{p_k}{|\mathcal{F}_k|N} \mathrm{tr} \bigg[  \boldsymbol{\hat{g}}_k   \boldsymbol{\hat{g}}_k^* \boldsymbol{\Omega}_k^{'} \bigg] & \stackrel{\text{(a)}}{\approx}
    \frac{p_k}{|\mathcal{F}_k|N} \mathrm{tr} \big[  \boldsymbol{\Gamma}_k^{(g)} \boldsymbol{\Omega} \big] \\
     & \stackrel{\text{(b)}}{\approx}
    \frac{p_k}{|\mathcal{F}_k|N} \mathrm{tr} \big[ \boldsymbol{\Gamma}_k^{(g)} \Tmat_k \big]. \label{eq:SINRAssyMMSE}
\end{align}
where (a) follows from \cite[Lemmas 4 and 6]{6172680} and (b) is obtained after applying Th. \ref{th:RMT1} by substituting $\boldsymbol{{D}}= \boldsymbol{\Gamma}_k^{(g)} \, p_k$, (\emph{ii}) $\boldsymbol{{R}}_j=\boldsymbol{\Gamma}_j^{(g)} \, p_k$, and (\emph{iii}) $\boldsymbol{{S}} + z\boldsymbol{{I}}_M=\frac{1}{|\mathcal{F}_k|N} \boldsymbol{\Sigma}_k$ while $\Tmat_k$ is defined next
\begin{equation}
    \Tmat_k = \bigg( \frac{1}{|\mathcal{F}_k|N} \sum \limits_{i = 1 }^K \frac{\boldsymbol{\Gamma}_i^{(g)}}{1+e_{i}} \, p_i  + \frac{1}{|\mathcal{F}_k|N} \boldsymbol{\Sigma}_k \bigg)^{\!-1}.
\end{equation}
The necessary coefficients can be calculated as $e_{j} = \lim_{n \to \infty} e_{j}^{(n)} $ with
\begin{align}
    e_{j}^{(n)} & = \frac{p_j}{|\mathcal{F}_j|N} \, \mathrm{tr} \Bigg[ \boldsymbol{\Gamma}_j^{(g)} \bigg( \frac{1}{|\mathcal{F}_j|N} \sum \limits_{i = 1 }^K \frac{\boldsymbol{\Gamma}_i^{(g)}}{1+e_{i}} \, p_i  + \frac{1}{|\mathcal{F}_j|N} \boldsymbol{\Sigma}_k \bigg)^{\!-1}  \Bigg].
\end{align}
The fixed-point algorithm can be used to compute $e_{j}^{(n)}$ and
has been proved to converge \cite{6172680}. Finally, given that all the involved matrices in $\overline{\mathrm{SINR}}_k$ are block-diagonal, i.e. $\Tmat_k = \mathrm{diag}\{ \boldsymbol{T}_{m,k} \ls \mathrm{for} \ls m\in \mathcal{F}_k \}$ the expression in \eqref{eq:SINRAppr} is obtained where $\boldsymbol{T}_{m,k}$ is defined in \eqref{eq:Tmk}.

\section{Proof of Th. \ref{Th:T2}}\label{proof:SINRAssymDL}
From Eq. \eqref{eq:SINRRZF}, we can derive an approximation for each of the terms in the numerator and denominator, respectively. In order not to overload the formulation, we will denote by $\boldsymbol{\bf \hat{g}}_k = \boldsymbol{m}_k^{(s)} \circ \boldsymbol{\hat{g}}_k$ the sparse version of the channel. We also define $\boldsymbol{\Omega} = \big[  \boldsymbol{\bf \hat{G}} \boldsymbol{\bf \hat{G}}^* + \rho \boldsymbol{I} \big]^{-1} = \frac{1}{MN} \boldsymbol{\Omega}^{'}$ with $\boldsymbol{\Omega}^{'} = \big[  \frac{1}{MN}\boldsymbol{\bf \hat{G}} \boldsymbol{\bf \hat{G}}^* + \frac{\rho}{MN} \boldsymbol{I} \big]^{-1}$. Denote by $\boldsymbol{\Omega}_k $ and $\boldsymbol{\Omega}_k^{'}$ the same as $\boldsymbol{\Omega}$ and $\boldsymbol{\Omega}^{'}$ after removing the contribution of UE $k$ (the same applies to $\boldsymbol{\Omega}_{k,i}$ where the contributions of UEs $k$ and $i$ are removed). We first calculate the value of $\lambda_k$, ensuring that $\mathbb{E} \{  ||\boldsymbol{W} \boldsymbol{v}_k||^2 \} = 1$. 
\begin{align}
    \lambda_k = \frac{1}{\sqrt{ \mathbb{E}\{ \boldsymbol{\bf \hat{g}}_k^*\boldsymbol{\Omega} \boldsymbol{W}^* \boldsymbol{W} \boldsymbol{\Omega} \boldsymbol{\bf \hat{g}}_k\} }}
\end{align}
The term inside the squared root can be asymptotically approximated for large $MN$, $K$ as follows:
\begin{align}
     \boldsymbol{\bf \hat{g}}_k^*\boldsymbol{\Omega} \boldsymbol{W}^* \boldsymbol{W}  \boldsymbol{\Omega}   \boldsymbol{\bf \hat{g}}_k & = \frac{\boldsymbol{\bf \hat{g}}_k^*\boldsymbol{\Omega}_k \boldsymbol{W}^* \boldsymbol{W}  \boldsymbol{\Omega}_k \boldsymbol{\bf \hat{g}}_k}{(1 + \boldsymbol{\bf \hat{g}}_k^*\boldsymbol{\Omega}_k   \boldsymbol{\bf \hat{g}}_k )^2} \\
     & \stackrel{\text{(a)}}{\approx} \frac{ \frac{1}{(MN)^2}  \mathrm{tr}\big[ \boldsymbol{\Gamma}_k^{(g)} \boldsymbol{\Omega}_k^{'} \boldsymbol{W}^* \boldsymbol{W}  \boldsymbol{\Omega}_k^{'} \big] }{(1 + \frac{1}{MN} \mathrm{tr}\big[ \boldsymbol{\Gamma}_k^{(g)} \boldsymbol{\Omega}_k^{'} \big])^2} \\
     & \stackrel{\text{(b)}}{\approx} \frac{ \frac{1}{(MN)^2}  \mathrm{tr}\big[ \boldsymbol{\Gamma}_k^{(g)} \boldsymbol{T}^{'}(\frac{\rho}{MN},\boldsymbol{W}^* \boldsymbol{W}  ) \big] }{(1 + \frac{1}{MN} \mathrm{tr}\big[ \boldsymbol{\Gamma}_k^{(g)} \boldsymbol{T} \big])^2} \\
     & \stackrel{\text{(c)}}{=} \frac{\delta_k}{(1 + \mu_k)^2}
\end{align}
where (a)  is obtained using \cite[Lemma 4]{6172680} and that  $\boldsymbol{\Omega}_k = \frac{1}{MN} \boldsymbol{\Omega}_k^{'}$, (b) results from \cite[Lemma 6]{6172680} and applying  Th. \ref{Th:T2} and Th. \ref{Th:T1} in the numerator and denominator, respectively, with $\boldsymbol{D} = \boldsymbol{\Gamma}_k^{(g)}$, $\boldsymbol{\Phi} = \boldsymbol{W}^* \boldsymbol{W}$, $\boldsymbol{S} = \boldsymbol{0}$, $z = \frac{\rho}{MN}$. Finally, (c) defines the values of $\delta_k= \frac{1}{(MN)^2}  \mathrm{tr}\big[ \boldsymbol{\Gamma}_k^{(g)} \boldsymbol{T}^{'}(\frac{\rho}{MN},\boldsymbol{W}^* \boldsymbol{W}) \big]$ and $\mu_k = \frac{1}{MN} \mathrm{tr}\big[ \boldsymbol{\Gamma}_k^{(g)} \boldsymbol{T} \big]$ as they will be repeatedly used later. As a consequence, from the continous mapping theorem:
\begin{align}\label{eq:lambdak}
    \lambda_k \approx \frac{1}{\sqrt{\frac{\delta_k}{(1 + \mu_k)^2}}}
\end{align}
For the numerator of \eqref{eq:SINRRZF}, given by $|\mathbb{E}\{\boldsymbol{g}_k^{*} \boldsymbol{v}_k   \}  |^2$, we can compute an approximated deterministic equivalent for the term inside the expectation in a similar manner as for $\lambda_k$:
\begin{align}
    \boldsymbol{g}_k^{*} \boldsymbol{v}_k & = \lambda_k \boldsymbol{g}_k^*\boldsymbol{\Omega}   \boldsymbol{\bf \hat{g}}_k  \\
    & \stackrel{\text{(a)}}{=} \lambda_k \frac{\boldsymbol{g}_k^*\boldsymbol{\Omega}_k   \boldsymbol{\bf \hat{g}}_k }{1 + \boldsymbol{g}_k^*\boldsymbol{\Omega}_k   \boldsymbol{\bf \hat{g}}_k} \\
    & \stackrel{\text{(b)}}{\approx}  \lambda_k \frac{ \frac{1}{MN}  \mathrm{tr}\big[ \boldsymbol{\Gamma}_k^{(g)} \boldsymbol{\Omega}_k^{'} \big] }{1 + \frac{1}{MN} \mathrm{tr}\big[ \boldsymbol{\Gamma}_k^{(g)} \boldsymbol{\Omega}_k^{'} \big]} \\
    & \stackrel{\text{(c)}}{\approx} \lambda_k \frac{\mu_k}{1 + \mu_k}
\end{align}
where (a) follows from \cite[Lemma 1]{6172680} (b) is derived applying \cite[Lemma 4]{6172680} and the fact that  $\boldsymbol{\Omega}_k = \frac{1}{MN} \boldsymbol{\Omega}_k^{'}$. Finally, (c) is obtained by applying the definition of $\mu_k$ previously derived. From the continuous mapping theorem and substituting the value of $\lambda_k$ provided in \eqref{eq:lambdak}, the numerator therefore has an approximated value of
\begin{align}
    |\mathbb{E}\{\boldsymbol{g}_k^{*} \boldsymbol{v}_k   \}  |^2 & \approx \lambda_k^2 \frac{\mu_k^2}{(1+\mu_k)^2} \\
    & = \frac{\mu_k^2}{\delta_k}
\end{align}

For the interfering terms $\mathbb{E}\{ |\boldsymbol{g}_k^{*} \boldsymbol{v}_i |^2  \}  $ we can proceed similarly and obtain a deterministic approximation by considering the term inside the expectation as follows:
\begin{align}
    |\boldsymbol{g}_k^{*} \boldsymbol{v}_i|^2 &  = \lambda_i^2 |\boldsymbol{g}_k^{*}  \boldsymbol{\Omega}   \boldsymbol{\bf \hat{g}}_i  |^2 \\
    & \stackrel{\text{(a)}}{=} \lambda_i^2 \frac{|\boldsymbol{g}_k^{*}  \boldsymbol{\Omega}_i   \boldsymbol{\bf \hat{g}}_i  |^2}{(1+ \boldsymbol{\bf \hat{g}}_i^*  \boldsymbol{\Omega}_i   \boldsymbol{\bf \hat{g}}_i )^2}\\ 
    & \stackrel{\text{(b)}}{=}  \lambda_i^2\frac{| \frac{1}{MN} \boldsymbol{g}_k^{*}  \boldsymbol{\Omega}_i^{'}   \boldsymbol{\bf \hat{g}}_i  |^2}{(1+ \frac{1}{MN} \boldsymbol{\bf \hat{g}}_i^*  \boldsymbol{\Omega}_i^{'}   \boldsymbol{\bf \hat{g}}_i )^2} \\
    & \stackrel{\text{(c)}}{\approx} \lambda_i^2 \frac{|\frac{1}{MN} \boldsymbol{g}_k^{*}  \boldsymbol{\Omega}_i^{'}   \boldsymbol{\bf \hat{g}}_i  |^2}{(1+ \mu_i )^2} \\
    & \stackrel{\text{(d)}}{=} \frac{1}{\delta_i} |\frac{1}{MN} \boldsymbol{g}_k^{*}  \boldsymbol{\Omega}_i^{'}   \boldsymbol{\bf \hat{g}}_i  |^2
\end{align}
where (a) follows from \cite[Lemma 1]{6172680}, (b) substitutes $\boldsymbol{\Omega}_i = \frac{1}{MN} \boldsymbol{\Omega}_i^{'}$, (c) applies the definition of $\mu_k$ in the denominator and (d) substitutes the value of $\lambda_i$ previously derived.

To get a deterministic equivalent for the previous equation, we first know that:
\begin{align}\label{eq:intki}
    |\frac{1}{MN} \boldsymbol{g}_k^{*}  \boldsymbol{\Omega}_i^{'}  \boldsymbol{\bf \hat{g}}_i  |^2 \approx \frac{1}{(MN)^2} \boldsymbol{g}_k^{*} \boldsymbol{\Omega}_i^{'} \boldsymbol{\Gamma}_i^{(g)}  \boldsymbol{\Omega}_i^{'} \boldsymbol{g}_k
\end{align}
being a direct consequence of  \cite[Lemma 4]{6172680}. After applying the matrix inversion lemma to $\boldsymbol{\Omega}_i^{'}$ to remove the dependency with respect to UE $k$, we obtain that
\begin{align}\label{eq:matSigmaki}
    \boldsymbol{\Omega}_i^{'} = \boldsymbol{\Omega}_{i,k}^{'} - \frac{ \frac{1}{MN} \boldsymbol{\Omega}_{i,k}^{'} \boldsymbol{ \hat{g}}_k \boldsymbol{ \hat{g}}_k^*  \boldsymbol{\Omega}_{i,k}^{'}}{1 +  \frac{1}{MN} \boldsymbol{ \hat{g}}_k^* \boldsymbol{\Omega}_{i,k}^{'}\boldsymbol{ \hat{g}}_k } 
\end{align}
Substituting \eqref{eq:matSigmaki} in \eqref{eq:intki} yields the following:
\begin{align}
    \frac{1}{(MN)^2} \boldsymbol{g}_k^{*} \boldsymbol{\Omega}_i^{'} \boldsymbol{\Gamma}_i^{(g)}  \boldsymbol{\Omega}_i^{'} \boldsymbol{g}_k = \mathrm{T}_1 + \mathrm{T}_2 + \mathrm{T}_3
\end{align}
where each of the terms is provided below:
\begin{align}
    \mathrm{T}_1 & = \frac{1}{(MN)^2} \boldsymbol{g}_k^{*} \boldsymbol{\Omega}_{i,k}^{'} \boldsymbol{\Gamma}_i^{(g)}  \boldsymbol{\Omega}_{i,k}^{'} \boldsymbol{g}_k \\
    &  \stackrel{\text{(a)}}{\approx} \frac{1}{(MN)^2}  \mathrm{tr}\big[ \boldsymbol{R}_k^{(g)} \boldsymbol{T}^{'}(\frac{\rho}{MN},\boldsymbol{\Gamma}_i^{(g)}) \big]
\end{align}
where (a) combines both \cite[Lemma 4]{6172680}  and Th. \ref{Th:T2} with the following substitutions  $\boldsymbol{D} = \boldsymbol{\Gamma}_k^{(g)}$, $\boldsymbol{\Phi} = \boldsymbol{\Gamma}_i^{(g)}$, $\boldsymbol{S} = \boldsymbol{0}$, $z = \frac{\rho}{MN}$. In addition,

\begin{align}
    \mathrm{T}_2 & =\frac{1}{(MN)^2} \frac{ \frac{1}{(MN)^2} |\boldsymbol{\hat{g}}_k^{*} \boldsymbol{\Omega}_{i,k}^{'} \boldsymbol{g}_k  |^2 \boldsymbol{\hat{g}}_k^{*} \boldsymbol{\Omega}_{i,k}^{'}\boldsymbol{\Gamma}_i^{(g)} \boldsymbol{\Omega}_{i,k}^{'}\boldsymbol{\hat{g}}_k }{( 1 + \frac{1}{MN} \boldsymbol{\hat{g}}_k^{*} \boldsymbol{\Omega}_{i,k}^{'}\boldsymbol{\hat{g}}_k )^2} \\
    & \stackrel{\text{(a)}}{\approx} \frac{1}{MN} \frac{\mu_k^2 \frac{1}{MN}  \boldsymbol{\hat{g}}_k^{*} \boldsymbol{\Omega}_{i,k}^{'}\boldsymbol{\Gamma}_i^{(g)} \boldsymbol{\Omega}_{i,k}^{'}\boldsymbol{\hat{g}}_k }{(1 + \mu_k)^2} \\
    & \stackrel{\text{(b)}}{\approx} \frac{1}{MN} \frac{\mu_k^2 \frac{1}{MN} \mathrm{tr}\big[ \boldsymbol{\Gamma}_k^{(g)} \boldsymbol{T}^{'}(\frac{\rho}{MN},\boldsymbol{\Gamma}_i^{(g)}) \big] }{(1 + \mu_k)^2}
\end{align}
where (a) comes from the definition of $\mu_i$, and (b) arises from applying \cite[Lemma 6]{6172680} and Th. \ref{Th:T2} to the term $\frac{1}{MN}  \boldsymbol{\hat{g}}_k^{*} \boldsymbol{\Omega}_{i,k}^{'}\boldsymbol{\Gamma}_i^{(g)} \boldsymbol{\Omega}_{i,k}^{'}\boldsymbol{\hat{g}}_k$ with the same substitutions as for $\mathrm{T}_1$. Finally, the last term can be computed as

\begin{align}
    \mathrm{T}_3 & = - \frac{2}{(MN)^2}\mathbb{R} \bigg\{ \frac{\frac{1}{MN} \boldsymbol{\hat{g}}_k^{*}\boldsymbol{\Omega}_{i,k}^{'}\boldsymbol{{g}}_k \boldsymbol{{g}}_k^*  \boldsymbol{\Omega}_{i,k}^{'} \boldsymbol{\Gamma}_i^{(g)}  \boldsymbol{\Omega}_{i,k}^{'} \boldsymbol{\hat{g}}_k  }{1 +  \frac{1}{MN} \boldsymbol{ \hat{g}}_k^* \boldsymbol{\Omega}_{i,k}^{'}\boldsymbol{ \hat{g}}_k} \bigg\} \\
    & \stackrel{\text{(a)}}{\approx} \frac{2}{MN}\mathbb{R} \bigg\{ \frac{\mu_k \frac{1}{MN} \boldsymbol{{g}}_k^*  \boldsymbol{\Omega}_{i,k}^{'} \boldsymbol{\Gamma}_i^{(g)}  \boldsymbol{\Omega}_{i,k}^{'} \boldsymbol{\hat{g}}_k  }{1 +  \mu_k} \bigg\} \\
    & \stackrel{\text{(b)}}{\approx} \frac{2}{MN}\mathbb{R} \bigg\{ \frac{\mu_k \frac{1}{MN} \mathrm{tr}\big[ \boldsymbol{\Gamma}_k^{(g)} \boldsymbol{T}^{'}(\frac{\rho}{MN},\boldsymbol{\Gamma}_i^{(g)})  }{1 +  \mu_k} \bigg\}
\end{align}
where (a) is obtained from the definition of $\mu_k$ and (b) follows the same step as to calculate $\mathrm{T}_2$ (b). Consequently, the interfering terms accept an assymptotic approximation as follows:
\begin{align}
    |\mathbb{E}\{\boldsymbol{g}_k^{*} \boldsymbol{v}_k   \}  |^2 & \approx \frac{\theta_{k,i}}{\delta_i}
\end{align}
with $\theta_{k,i} =  \mathrm{T}_1 + \mathrm{T}_2 + \mathrm{T}_3$.

Finally, the term $\mathrm{var}( \boldsymbol{g}_k^{*} \boldsymbol{v}_k )$ can be shown to approximately converge to zero in the asymptotic regime as follows:
\begin{align}
    \mathrm{var}( \boldsymbol{g}_k^{*} \boldsymbol{v}_k ) & = \mathbb{E}\{ |\boldsymbol{g}_k^{*} \boldsymbol{v}_k|^2 \} - \mathbb{E}\{ \boldsymbol{g}_k^{*} \boldsymbol{v}_k \}^2 \\
    & \approx  \bigg( \lambda_k \frac{\mu_k}{1 + \mu_k} \bigg)^2 - \bigg( \lambda_k \frac{\mu_k}{1 + \mu_k} \bigg)^2 
\end{align}
As a consequence, the result in Th. \ref{Th:T2} is obtained.

\section{Proof of Prop. \ref{prop:WMMSE}}\label{proof:WMMSE}
From Eq. \eqref{eq:SINRmmse}, under perfect CSI it can be shown that the SINR achieved by UE $k$ is:
\begin{align}
    \mathrm{SINR}_k = \boldsymbol{{g}}_k^{*} \bigg(  \sum 
    \limits_{i\neq k}^{K} \boldsymbol{{g}}_i \boldsymbol{{g}}_i^{*}p_i + \boldsymbol{\Sigma}_k \bigg)^{-1}  \boldsymbol{{g}}_k,
\end{align}
where, for simplicity we assume that $\Ms = \boldsymbol{1}$ though the same analysis and conclusion is valid for subsets of APs and UEs. Therefore, $\boldsymbol{\Sigma}_k$ is a block diagonal matrix $\boldsymbol{\Sigma}_k =  \mathrm{diag}\{ \boldsymbol{\Sigma}_{k,m} \in \mathbb{C}^{L\times L} \ls \mathrm{for} \ls m \in \mathcal{F}_k  \}  $ where $\boldsymbol{\Sigma}_{k,m} =  \sigma^2 \Wmh{} \Wm$. Note that $\boldsymbol{{g}}_k = \boldsymbol{W} \boldsymbol{{h}}_k$. As a consequence:
\begin{align}
    \mathrm{SINR}_k & = \boldsymbol{{h}}_k^{*} \boldsymbol{W} \bigg(  \sum 
    \limits_{i\neq k}^{K} \boldsymbol{W}^{*} \boldsymbol{{h}}_i \boldsymbol{{h}}_i^{*} \boldsymbol{W} p_i + \sigma^2 \boldsymbol{W}^{*}\boldsymbol{W} \bigg)^{-1}   \boldsymbol{W} \boldsymbol{{h}}_k \\
    & =  \boldsymbol{{h}}_k^{*} \boldsymbol{W} \bigg( \boldsymbol{W}^* \Big( \sum 
    \limits_{i\neq k}^{K}  \boldsymbol{{h}}_i \boldsymbol{{h}}_i^{*}  p_i + \sigma^2 \boldsymbol{I} \Big) \boldsymbol{W}  \bigg)^{-1}   \boldsymbol{W} \boldsymbol{{h}}_k.
\end{align}
Consider the generic case of $\mathrm{rank} ( \boldsymbol{W}_m ) = r_m \leq L$. It can be easily shown that if $\exists \ls r_m < L$ $\bigg( \boldsymbol{W}^* \Big( \sum 
\limits_{i\neq k}^{K}  \boldsymbol{{h}}_i \boldsymbol{{h}}_i^{*}  p_i + \sigma^2 \boldsymbol{I} \Big) \boldsymbol{W}  \bigg)^{-1}$ does not exist. As a consequence, each $\boldsymbol{W}_m$ must be full rank. After doing the compact SVD on $\boldsymbol{W} = \boldsymbol{U} \boldsymbol{Q} \in \mathbb{C}^{NM \times r}$ where $r = \sum_m r_m$ and both $\boldsymbol{U}$ and $\boldsymbol{Q}$ are block diagonal. More particularly, $\boldsymbol{U} = \mathrm{diag}\{ \boldsymbol{U}_m \ls \mathrm{for} \ls m=1,\dots,M \}$ with each $\boldsymbol{U}_m \in \mathbb{C}^{N \times r_m}$ and $\boldsymbol{U}_m^* \boldsymbol{U}_m = \boldsymbol{I}$. Similarly, $\boldsymbol{Q} = \mathrm{diag}\{ \boldsymbol{Q}_m \ls \mathrm{for} \ls m=1,\dots,M \}$ with each $\boldsymbol{Q}_m \in \mathbb{C}^{r_m \times r_m}$. Then it follows that
\begin{align}
\begin{gathered}
    \mathrm{SINR}_k  =  \\ \boldsymbol{{h}}_k^{*} \boldsymbol{U} \boldsymbol{Q} \bigg( \boldsymbol{Q}^{*} \boldsymbol{U}^{*} \Big( \sum 
    \limits_{i\neq k}^{K}   \boldsymbol{{h}}_i \boldsymbol{{h}}_i^{*}  p_i + \sigma^2  \boldsymbol{I} \Big) \boldsymbol{U}  \boldsymbol{Q}  \bigg)^{-1}   \boldsymbol{Q}^{*} \boldsymbol{U}^{*} \boldsymbol{{h}}_k \\
    %& = \boldsymbol{{h}}_k^{*} \boldsymbol{U} \boldsymbol{Q} \boldsymbol{Q}^{-1} \bigg( \boldsymbol{U}^{*} \Big( \sum 
    %\limits_{i\neq k}^{K}   \boldsymbol{{h}}_i \boldsymbol{{h}}_i^{*}  p_i + \sigma^2  \boldsymbol{I} \Big) \boldsymbol{U} \bigg)^{-1} \boldsymbol{Q}^{*-1}  \boldsymbol{Q}^{*} \boldsymbol{U}^{*} \boldsymbol{{h}}_k \\
     = \boldsymbol{{h}}_k^{*} \boldsymbol{U} \bigg(   \boldsymbol{U}^{*} \Big( \sum 
    \limits_{i\neq k}^{K}   \boldsymbol{{h}}_i \boldsymbol{{h}}_i^{*}  p_i + \sigma^2  \boldsymbol{I} \Big) \boldsymbol{U} \bigg)^{-1}  \boldsymbol{U}^{*} \boldsymbol{{h}}_k
    \end{gathered}
\end{align}
As a consequence, the UL SINR after MMSE reception under perfect CSI does not depend on $\boldsymbol{Q}$. Therefore, any non-singular $\boldsymbol{Q}$ maximizes $ \mathrm{SINR}_k$.

%%% Proof of Prop DL RZF Perfect CSI
\section{Proof of Prop. \ref{prop:WRZF}}\label{proof:WRZF}
Under perfect CSI, the DL-SINR under RZF precoding is
\begin{align}
    \mathrm{SINR}_k & =  \frac{|\boldsymbol{g}_k^{*} \boldsymbol{v}_k^{*}  |^2  p_k}{ \sum \limits_{i\neq k}^K |\boldsymbol{g}_k^{*} \boldsymbol{v}_i|^2p_i  + \sigma^2} 
\end{align}
According to \cite{6200372}, the term $\sum \limits_{i\neq 1}^K |\boldsymbol{g}_k^{*} \boldsymbol{v}_i|^2p_i  + \sigma^2 \approx \sum \limits_{i\neq k}^K |\boldsymbol{g}_i^{*} \boldsymbol{v}_k|^2p_i  + \sigma^2$. As a consequence, $\mathrm{SINR}_k$ can be approximately rewritten as 
\begin{align}\label{eq:SLNRSINR}
    \mathrm{SINR}_k & \approx  \frac{|\boldsymbol{g}_k^{*} \boldsymbol{v}_k^{*}  |^2  p_k}{ \sum \limits_{i\neq k}^K |\boldsymbol{g}_i^{*} \boldsymbol{v}_k|^2p_i  + \sigma^2} .
\end{align}
Again, and for simplicity, we assume $\Ms = \boldsymbol{1}$. Using a RZF precoding
\begin{align}
    \boldsymbol{V} 
    & = \big[  \boldsymbol{{G}} \boldsymbol{{G}}^* + \rho \boldsymbol{I} \big]^{-1}  \boldsymbol{{G}}\boldsymbol{\Lambda} .
\end{align}
where $\boldsymbol{\Lambda} = \mathrm{diag}( \lambda_1,\dots,\lambda_K)$ such that  $||\boldsymbol{v}_k||^2 = 1$.  Therefore
\begin{align}
    \lambda_k = \frac{1}{\sqrt{|| \boldsymbol{W} \big[  \boldsymbol{{G}} \boldsymbol{{G}}^* + \rho \boldsymbol{I} \big]^{-1} \boldsymbol{g}_k ||^2 }}.
\end{align}
Substituting the previous expression in Eq. \eqref{eq:SLNRSINR} we obtain
\begin{align}\label{eq:SINRApprx}
    \mathrm{SINR}_k & \approx \frac{|\boldsymbol{g}_k^{*} \big[  \boldsymbol{{G}} \boldsymbol{{G}}^* + \rho \boldsymbol{I} \big]^{-1} \boldsymbol{g}_k  |^2 \lambda_k^2 p_k}{ \sum \limits_{i\neq k}^K |\boldsymbol{g}_i^{*}  \big[  \boldsymbol{{G}} \boldsymbol{{G}}^* + \rho \boldsymbol{I} \big]^{-1} \boldsymbol{g}_k  |^2  p_i  + \sigma^2} \\
    %& = \frac{|\boldsymbol{g}_k^{*} \big[  \boldsymbol{{G}} \boldsymbol{{G}}^* + \rho \boldsymbol{I} \big]^{-1} \boldsymbol{g}_k  |^2  p_k}{ \sum \limits_{i\neq k}^K |\boldsymbol{g}_i^{*}  \big[  \boldsymbol{{G}} \boldsymbol{{G}}^* + \rho \boldsymbol{I} \big]^{-1} \boldsymbol{g}_k  |^2  p_i  + \sigma^2 \boldsymbol{g}_k^*  \big[  \boldsymbol{{G}} \boldsymbol{{G}}^* + \rho \boldsymbol{I} \big] \boldsymbol{W}^* \boldsymbol{W}  \big[  \boldsymbol{{G}} \boldsymbol{{G}}^* + \rho \boldsymbol{I} \big]  \boldsymbol{g}_k} \\
    & = \frac{|\boldsymbol{g}_k^{*} \big[  \boldsymbol{{G}} \boldsymbol{{G}}^* + \rho \boldsymbol{I} \big]^{-1} \boldsymbol{g}_k  |^2 p_k}{ \sum \limits_{i\neq k}^K |\boldsymbol{g}_i^{*}  \big[  \boldsymbol{{G}} \boldsymbol{{G}}^* + \rho \boldsymbol{I} \big]^{-1} \boldsymbol{g}_k  |^2 \lambda_k^2 p_i  + \frac{\sigma^2}{\lambda_k^2}} \\
    & = \frac{\boldsymbol{h}_k^{*} \boldsymbol{\Omega} \boldsymbol{h}_k \boldsymbol{h}_k^{*} \boldsymbol{\Omega} \boldsymbol{h}_k   p_k}{ \boldsymbol{h}_k^{*}  \boldsymbol{\Omega} \big( \boldsymbol{H}\boldsymbol{P} \boldsymbol{H}^*  -  \boldsymbol{h}_k\boldsymbol{h}_k^{*}p_k + \sigma^2 \boldsymbol{I} \big)  \boldsymbol{\Omega} \boldsymbol{h}_k}
\end{align}
where, in the last step we define by $\boldsymbol{\Omega} = \boldsymbol{W}  \big[  \boldsymbol{{G}} \boldsymbol{{G}}^* + \rho \boldsymbol{I} \big]^{-1} \boldsymbol{W}^*$. Now, by compact SVD $\boldsymbol{W} = \boldsymbol{U} \boldsymbol{Q} \in \mathbb{C}^{NM \times r}$ where $r = \sum_m r_m$ and both $\boldsymbol{U}$ and  $\boldsymbol{Q}$ are block diagonal. Let $\boldsymbol{\bf H} = \boldsymbol{U}^* \boldsymbol{H}$ and $\boldsymbol{\bf h}_k = \boldsymbol{U}^* \boldsymbol{h}_k$. Then $\boldsymbol{h}_k^{*} \boldsymbol{\Omega}$ can be written as
\begin{align}
    \boldsymbol{h}_k^{*} \boldsymbol{\Omega} & = \boldsymbol{h}_k^{*} \boldsymbol{W}  \big[ \boldsymbol{W}^* \boldsymbol{{H}} \boldsymbol{{H}}^* \boldsymbol{W} + \rho \boldsymbol{I} \big]^{-1} \boldsymbol{W}^* \\
    & =  \boldsymbol{\bf h}_k^{*} \big[  \boldsymbol{\bf H} \boldsymbol{\bf H}^*  + \rho (\boldsymbol{Q} \boldsymbol{Q}^*) ^{-1}\big]^{-1} \boldsymbol{U}^*.
\end{align}
We define by $\boldsymbol{B} = \big[  \boldsymbol{\bf H} \boldsymbol{\bf H}^*  + \rho (\boldsymbol{Q} \boldsymbol{Q}^*) ^{-1}\big]^{-1} $. Operating on \eqref{eq:SINRAppr}, we obtain that
\begin{align}\label{eq:SINRFinalU}
    \mathrm{SINR}_k & \approx \frac{ \boldsymbol{\bf h}_k^{*} \boldsymbol{B}\boldsymbol{\bf h}_k \boldsymbol{\bf h}_k^{*} \boldsymbol{B}\boldsymbol{\bf h}_k p_k}{ \boldsymbol{\bf h}_k^{*} \boldsymbol{B} \big( \boldsymbol{\bf H}\boldsymbol{P} \boldsymbol{\bf H}^* +  \sigma^2 \boldsymbol{I} \big)\boldsymbol{B} \boldsymbol{\bf h}_k  - \boldsymbol{\bf h}_k^{*} \boldsymbol{B}\boldsymbol{\bf h}_k \boldsymbol{\bf h}_k^{*} \boldsymbol{B}\boldsymbol{\bf h}_k p_k } \\
    & = \frac{\mathrm{K}_k}{1 - \mathrm{K}_k}
\end{align}
where $0 \leq \mathrm{K}_k \leq 1$ with $\mathrm{K}_k$ defined as
\begin{align}
    \mathrm{K}_k & = \frac{ \boldsymbol{\bf h}_k^{*} \boldsymbol{B}\boldsymbol{\bf h}_k \boldsymbol{\bf h}_k^{*} \boldsymbol{B}\boldsymbol{\bf h}_k p_k}{\boldsymbol{\bf h}_k^{*} \boldsymbol{B} \big( \boldsymbol{\bf H}\boldsymbol{P} \boldsymbol{\bf H}^* +  \sigma^2 \boldsymbol{I} \big)\boldsymbol{B} \boldsymbol{\bf h}_k} \\
    & = \frac{ \boldsymbol{b}_k^* \boldsymbol{\bf h}_k \boldsymbol{\bf h}_k^* \boldsymbol{b}_k p_k  }{ \boldsymbol{b}_k^* \big( \boldsymbol{\bf H}\boldsymbol{P} \boldsymbol{\bf H}^* +  \sigma^2 \boldsymbol{I} \big) \boldsymbol{b}_k }.
\end{align}
with $\boldsymbol{b}_k = \boldsymbol{B} \boldsymbol{\bf h}_k$.  Note that \eqref{eq:SINRFinalU} is an increasing function with respect to $\mathrm{K}_k$. Thus, maximizing $\mathrm{K}_k$ is equivalent to maximizing the SINR. Since $\mathrm{K}_k$ follows a Rayleigh quotient, the optimal $\boldsymbol{b}_k$  maximizing $\mathrm{K}_k$ is the eigenvector associated to the maximum eigenvalue of $\big( \boldsymbol{\bf H}\boldsymbol{P} \boldsymbol{\bf H}^* +  \sigma^2 \boldsymbol{I} \big)^{-1}\boldsymbol{\bf h}_k \boldsymbol{\bf h}_k^* $. Given that the previous matrix is rank-ones, there is only one eigenvector. As a consequence:
\begin{align}
    \boldsymbol{b}_k^{(\mathrm{max})} = \big( \boldsymbol{\bf H}\boldsymbol{P} \boldsymbol{\bf H}^* +  \sigma^2 \boldsymbol{I} \big)^{-1}\boldsymbol{\bf h}_k.
\end{align}
By definition, $\boldsymbol{b}_k = \boldsymbol{B} \boldsymbol{\bf h}_k = \big[  \boldsymbol{\bf H} \boldsymbol{\bf H}^*  + \rho (\boldsymbol{Q} \boldsymbol{Q}^*) ^{-1}\big]^{-1} \boldsymbol{\bf h}_k$. As a consequence, to obtain that  $\boldsymbol{b}_k = \boldsymbol{b}_k^{(\mathrm{max})}$, matrix $\boldsymbol{Q}$ has to satisfy $\boldsymbol{Q} \boldsymbol{Q}^* = \boldsymbol{I}$, i.e. being semi-unitary.

%% Proof of Prop of W
\section{Proof of Prop. \ref{prop:WAF}}\label{proof:WAF}
Let us assume a generic and nonsingular  $\boldsymbol{A}_m$. Then, $\boldsymbol{A}_m = \boldsymbol{U}_1 \boldsymbol{D}_1 \boldsymbol{V}_1^* $ with $\boldsymbol{U}_1$ and $\boldsymbol{V}_1$ being unitary. After adding $\boldsymbol{A}_m$, the output of the analog beamformer is 
\begin{align}
    \boldsymbol{\hat{W}}_m \boldsymbol{A}_m = \boldsymbol{\hat{W}}_m \boldsymbol{U}_1 \boldsymbol{D}_1 \boldsymbol{V}_1^*
\end{align}
Now, let us add the compensation matrix $\boldsymbol{F}_m$. Recall that the compensation matrix tries to somehow compensate the matrix that is in front of it, as shown in Eq. \eqref{eq:CompMat}.  In this case, for a generic $\boldsymbol{A}_m$, the compensation matrix of $\boldsymbol{\hat{W}}_m^*\boldsymbol{A}_m$ is  $\boldsymbol{F}_m = \boldsymbol{{V}}_1 \boldsymbol{{D}}_1^{-1} \boldsymbol{{V}}_1^{*}$ following Eq. \eqref{eq:CompMat}. Then, the product of the three matrices is:
\begin{align}
    \boldsymbol{\hat{W}}_m \boldsymbol{A}_m \boldsymbol{F}_m = \boldsymbol{\hat{W}}_m \boldsymbol{U}_1 \boldsymbol{D}_1 \boldsymbol{V}_1^* \boldsymbol{{V}}_1 \boldsymbol{{D}}_1^{-1} \boldsymbol{{V}}_1^{*} = \boldsymbol{\hat{W}}_m \boldsymbol{U}_1 \boldsymbol{{V}}_1^{*} 
\end{align}
Note that since both $\boldsymbol{{U}}_1$ and $\boldsymbol{{V}}_1$ are unitary, we are not modifying the optimality of the solution. As a consequence, $\boldsymbol{\hat{W}}_m \boldsymbol{A}_m \boldsymbol{F}_m$ is also an unconstrained combiner, as the initial one $\boldsymbol{\hat{W}}_m$, that does not change the output power.

\section{Proof of Prop. \ref{prop:UBLBHyb}}\label{proof:UBLBHyb}
For simplicity, let us assume $\Ms = \boldsymbol{1}$. Under perfect CSI and maximum ratio combining (MRC), i.e. $\boldsymbol{v}_k = \boldsymbol{g}_k$ for the unconstrained solution of $\boldsymbol{W}$, the SINR in \eqref{eq:SINRk} becomes
\begin{align}\label{eq:SINRk}
    \mathrm{SINR}_k = \frac{ |\boldsymbol{g}_k^* \boldsymbol{{g}}_k|^2 p_k }{ \sum
    \limits_{i\neq k}^{K}|\boldsymbol{g}_k^*  \boldsymbol{{g}}_i|^2 p_i + \sigma^2 \boldsymbol{g}_k^* \boldsymbol{g}_k }.
\end{align}
For $M\to\infty$ at a faster peace than $K$, $\boldsymbol{g}_k^*  \boldsymbol{{g}}_i \to 0$ almost surely. As a consequence, the asymptotic SINR achieved by UE $k$ is
\begin{align}\label{eq:SINRPCSI}
    \overline{\mathrm{SINR}}_k = \frac{p_k}{\sigma^2}  \sum \limits_{m =1}^M  \mathrm{tr} \Big[ \boldsymbol{W}_m \boldsymbol{R}_{m,k} \boldsymbol{W}_{m} \Big].
\end{align}
Let $\{ \lambda_{m,k}^{(1)},\dots,\lambda_{m,k}^{(N)} \}$ be the eigenvalues of $\boldsymbol{R}_{m,k}$ sorted in descending order. Recall that $\boldsymbol{W}_m$ is semi-unitary. Therefore we can construct a unitary $\boldsymbol{W}_m^{(u)} = [ \boldsymbol{W}_m \ls \boldsymbol{W}_{m,0} ]$ such that $\boldsymbol{W}_{m,0}^* \boldsymbol{W}_{m,0} = \boldsymbol{I}$ and $\boldsymbol{W}_{m}^* \boldsymbol{W}_{m,0} = \boldsymbol{0}$. Provided that $\boldsymbol{W}_m^{(u)}$ is unitary, $\boldsymbol{W}_m^{(u)\ls *} \boldsymbol{R}_{m,k} \boldsymbol{W}_m^{(u)}$ has the same eigenvalues as $\boldsymbol{R}_{m,k}$ and can be written as
\begin{align}
    \boldsymbol{W}_m^{(u) \ls *} \boldsymbol{R}_{m,k} \boldsymbol{W}_m^{(u)} & = \\
    & \left[ {\begin{array}{cc}
    \boldsymbol{W}_m^{*} \boldsymbol{R}_{m,k} \boldsymbol{W}_m & \boldsymbol{W}_m^{*} \boldsymbol{R}_{m,k} \boldsymbol{W}_{m,0} \\
    \boldsymbol{W}_{m,0} \boldsymbol{R}_{m,k} \boldsymbol{W}_m^{*} & \boldsymbol{W}_{m,0}^{ *} \boldsymbol{R}_{m,k} \boldsymbol{W}_{m,0} \\
  \end{array} } \right]
\end{align}
Denote the eigenvalues of $\boldsymbol{W}_m^{*} \boldsymbol{R}_{m,k} \boldsymbol{W}_m$ by $\mu_{m,k}^{(1)} \geq \mu_{m,k}^{(2)} \geq \dots \geq \mu_{m,k}^{(L)}$. For a fully digital receiver, the asymptotic SINR is
\begin{align}
    \overline{\mathrm{SINR}}_k^{\mathrm{FD}} = \frac{p_k}{\sigma^2}  \sum \limits_{m =1}^M  \mathrm{tr} \Big[ \boldsymbol{R}_{m,k} \Big].
\end{align}
By the Cauchy's interlacing theorem, the eigenvalues of the leading principal submatrix $\boldsymbol{W}_m^{*} \boldsymbol{R}_{m,k} \boldsymbol{W}_m$ satisfy
\begin{align}
    \lambda_{m,k}^{(i)} \geq \mu_{m,k}^{(i)} \geq \lambda_{m,k}^{(N-L+i)} \mspace{10mu} \mathrm{for} \mspace{4mu} i=1,\dots,L.
\end{align}
As a consequence, two bounds can be derived. A lower bound for the gap between hybrid and full digital occurs when $\mu_{m,k}^{(i)} = \lambda_{m,k}^{(i)}$. As a consequence, such a gap, denoted by $\delta_{\mathrm{LB}}$ is:
\begin{align}
    \delta & = \overline{\mathrm{SINR}}_k^{\mathrm{FD}} - \overline{\mathrm{SINR}}_k \\
    & = \frac{p_k}{\sigma^2}  \sum \limits_{m =1}^M \bigg( \sum \limits_{n=1}^N \lambda_{m,k}^{(n)} - \sum \limits_{n=1}^L \mu_{m,k}^{(n)} \bigg) \\
    & \geq \frac{p_k}{\sigma^2}  \sum \limits_{m =1}^M \sum \limits_{n=L+1}^N \lambda_{m,k}^{(n)} \\
    & = \delta_{\mathrm{LB}}.
\end{align}
To the contrary, the gap is maximum when $\mu_{m,k}^{(i)} = \lambda_{m,k}^{(N-M+i)}$. As a consequence, an upper bound on the gap between hybrid and full digital can be derived
\begin{align}
    \delta & = \overline{\mathrm{SINR}}_k^{\mathrm{FD}} - \overline{\mathrm{SINR}}_k \\
    & = \frac{p_k}{\sigma^2}  \sum \limits_{m =1}^M \bigg( \sum \limits_{n=1}^N \lambda_{m,k}^{(n)} - \sum \limits_{n=1}^L \mu_{m,k}^{(n)} \bigg) \\
    & \leq \frac{p_k}{\sigma^2}  \sum \limits_{m =1}^M \bigg( \sum \limits_{n=1}^L ( \lambda_{m,k}^{(n)} - \lambda_{m,k}^{(N-L + n)} ) +  \sum \limits_{n=L+1}^N \lambda_{m,k}^{(n)} \bigg) \\
    & = \delta_{\mathrm{UB}}
\end{align}

\end{document}